\documentclass[11pt]{article}

\usepackage{amsmath,amsfonts,amssymb,amsthm}
\usepackage{mathrsfs}
\usepackage{geometry}
\usepackage[numbers]{natbib}
\usepackage[colorlinks,citecolor=green,linkcolor=blue,bookmarks=false,hypertexnames=true]{hyperref}
\usepackage{graphicx}
\usepackage{tikz, tikz-cd}
\usepackage{float}

\newtheorem{definition}{Definition}
\newtheorem{theorem}{Theorem}
\newtheorem{remark}{Remark}
\newtheorem{lemma}{Lemma}
\newtheorem{assumption}{Assumption}

\newcommand{\const}{\text{const}}

\usepackage{titling}
\usepackage{array}
\preauthor{\begin{center}
    \large \lineskip .75em%
        \begin{tabular}[t]{>{\centering\arraybackslash}p{.45\textwidth}}}
        \postauthor{\end{tabular}\par\end{center}}
\makeatletter
\renewcommand\and{
  \end{tabular}%
  \hfill
  \begin{tabular}[t]{>{\centering\arraybackslash}p{.45\textwidth}}}
\makeatother

\begin{document}

\title{Dual Euler--Poincar\'e/Lie--Poisson formulation of subinertial stratified thermal ocean flow with identification of Casimirs as Noether quantities}

\author{F.J.\ Beron-Vera\\ Department of Atmospheric Sciences\\ Rosenstiel School of Marine, Atmospheric \& Earth Science\\ University of Miami\\ Miami, Florida, USA\\ fberon@miami.edu \and E.\ Luesink\\ Korteweg-De Vries Institute\\ University of Amsterdam\\ Amsterdam, The Netherlands\\ e.luesink@uva.nl}

\date{Started: May 9, 2024. This version: \today.\vspace{-0.25in}}

\maketitle

\begin{abstract} 
  This paper investigates the geometric structure of a quasigeostrophic approximation to a recently introduced reduced-gravity thermal rotating shallow-water model that accounts for stratification. Specifically, it considers a low-frequency approximation of a model for flow above the ocean thermocline, governed by primitive equations with buoyancy variations in both horizontal and vertical directions. Like the thermal model, the stratified variant generates circulation patterns reminiscent of submesoscale instabilities visible in satellite images. An improvement is its ability to model mixed-layer restratification due to baroclinic instability.
    
  The primary contribution of this paper is to demonstrate that the model is derived from an Euler--Poincaré variational principle, culminating in a Kelvin--Noether theorem, previously established solely for the primitive-equation parent model. The model's Lie--Poisson Hamiltonian structure, earlier obtained through direct calculation, is shown to result from a Legendre transform with the associated geometry elucidated by identifying the relevant momentum map.
    
  Another significant contribution of this paper is the identification of the Casimirs of the Lie--Poisson system, including a newly found weaker Casimir family forming the kernel of the Lie--Poisson bracket, which results in potential vorticity evolution independent of buoyancy details as it advects under the flow. These conservation laws related to particle relabeling symmetry are explicitly linked to Noether quantities from the Euler--Poincar\'e principle when variations are not constrained to vanish at integration endpoints.
    
  The dual Euler--Poincar\'e/Lie--Poisson formalism provides a unified framework for describing quasigeostrophic reduced-gravity stratified thermal flow, mirroring the approach used in the primitive-equation setting.
\end{abstract}

\tableofcontents

\section{Introduction}

As the upper ocean absorbs heat from a warming troposphere due to anthropogenic activities, an increase in lateral buoyancy gradients is expected. Recent research \cite{Holm-etal-21, Beron-21-POFa, Beron-24-POFa} has highlighted how the misalignment between the lateral gradient of buoyancy (temperature) and that of the mixed-layer thickness contributes to the proliferation of submesoscale (1–10 km) circulations (Fig.\@~\ref{fig:kelvin}, left panel). Commonly observed in satellite ocean color images (Fig.\@~\ref{fig:kelvin}, right panel), such circulations are manifestations of thermal instabilities \cite{Gouzien-etal-17}, ageostrophic phenomena consequential for turbulent transport and energy dissipation \cite{Fu-Ferrari-09}. These are characterized by a cascade of Kelvin-Helmholtz-like vortices which roll up along fronts, phenomenon that resembles the stretching and folding observed in Rayleigh–B\'enard convection in incompressible Euler--Boussinesq flow on a vertical plane \cite{Holm-Wei-23}.

\paragraph{Thermal ocean modeling.}

The modeling framework for the above theoretical development is provided by the \emph{thermal} rotating shallow-water equations, also known as \emph{Ripa's equations} \cite{Obrien-Reid-67, Schopf-Cane-83, Ripa-GAFD-93, Ripa-JFM-95, Dellar-03}. The rotating shallow-water equations represent a paradigm for ocean dynamics on timescales longer than a few hours \cite{Zeitlin-18}. These equations are derived by vertically integrating the primitive equations (PE), specifically the hydrostatic Euler-Boussinesq equations with Coriolis force, for a homogeneous layer (HL) of fluid, where the horizontal velocity is replaced by its vertical average. Using the notation introduced in \cite{Ripa-GAFD-93}, we refer to this model as HLPE.

The thermal rotating shallow-water equations follow similarly but start from the PE for a horizontally \emph{inhomogeneous} layer (IL) of fluid, where the horizontal velocity is assumed to be vertically uniform. As in \cite{Ripa-JFM-95}, we refer to this model as IL$^0$PE, where the superscript indicates that both horizontal velocity and density are depth independent. Unlike HLPE, the IL$^0$PE in a \emph{reduced-gravity setting}---where the layer of active fluid is bounded above by a rigid lid and below by a soft interface with an infinitely deep, inert abyss---allows for a simplified representation of upper-ocean dynamics \emph{and} thermodynamics, as this model can accommodate heat and freshwater fluxes across the ocean surface.  Furthermore, the low-frequency-dominant thermal-wind balance allows the IL$^0$PE to implicitly include vertical velocity shear, enabling the model to partially represent baroclinic instability.

\begin{figure}[t!]
  \centering%
  \includegraphics[width=.49\textwidth]{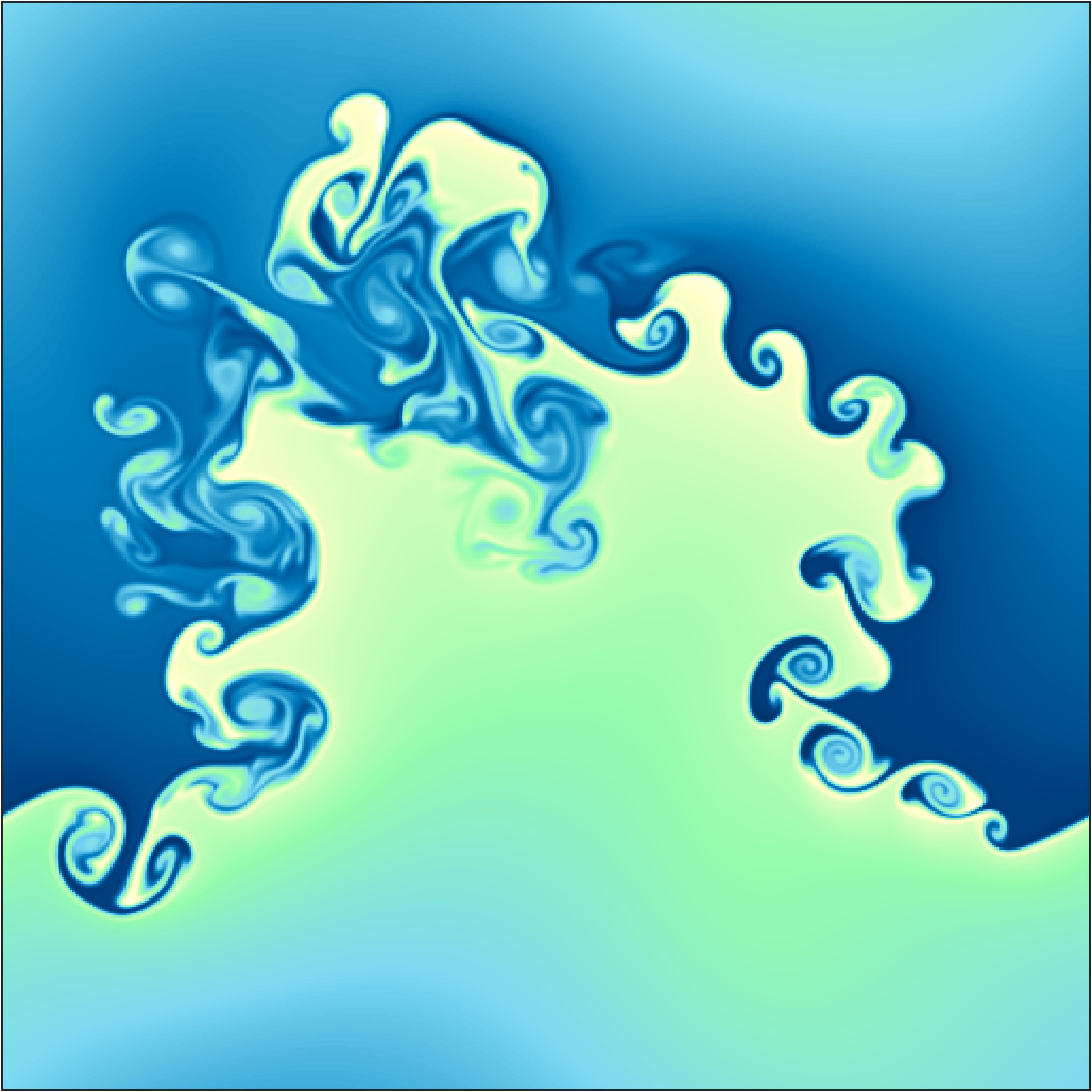}\,
  \includegraphics[width=.49\textwidth]{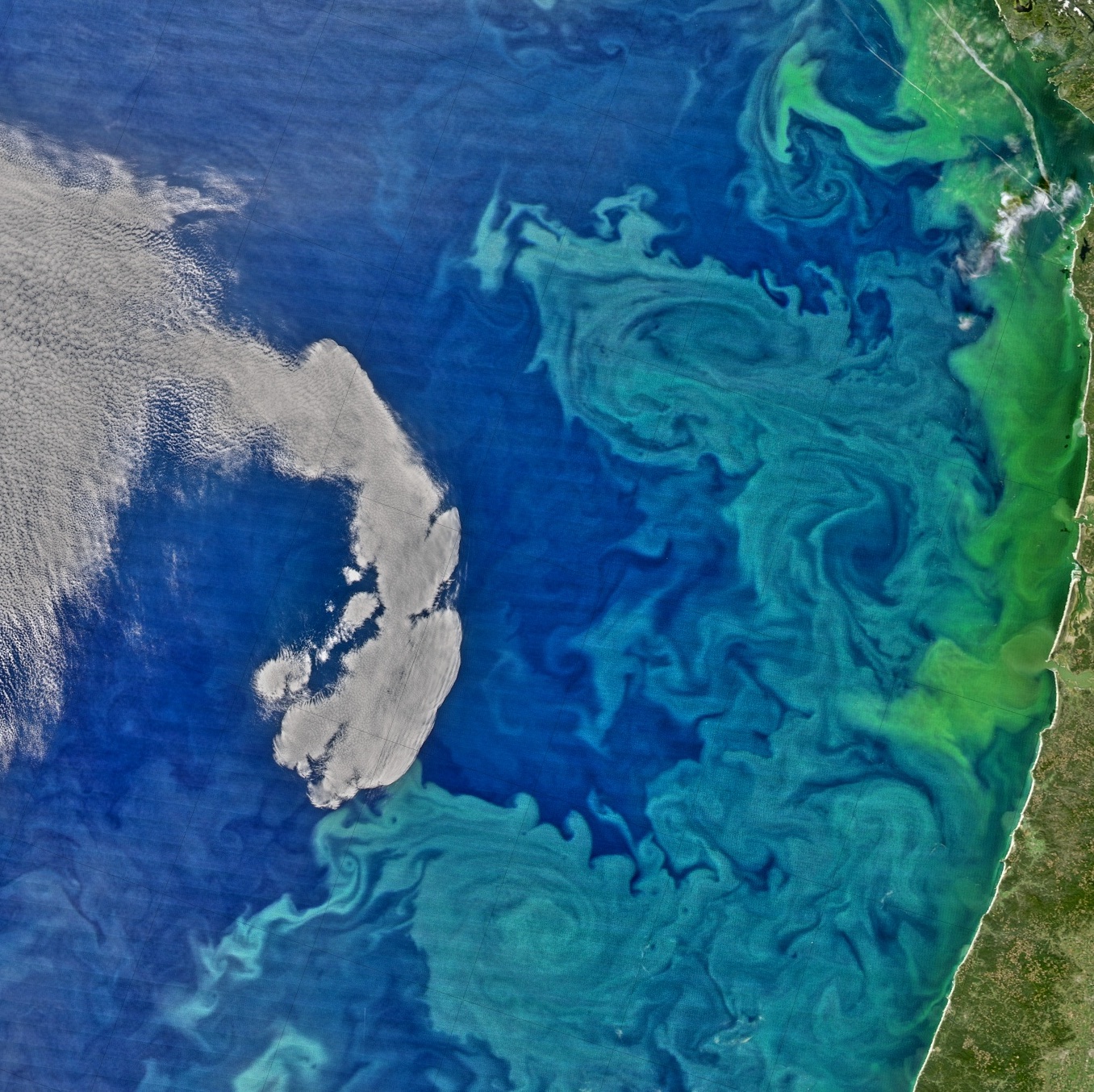}
  \caption{(left) Emerging Kelvin--Helmholtz-like vortices rolling up along a density front in a direct numerical simulation of the IL$^0$QG in a doubly periodic domain of size $R \approx 25$ km, corresponding to the (baroclinic) Rossby radius of deformation.  These vortices are low-frequency manifestations of inherently ageostrophic thermal instabilities.  (right) View of similar phytoplankton patterns on 30 March 2016 in the northeastern Pacific as composed using data acquired by the MODIS (Moderate Resolution Imaging Spectroradiometer) sensor mounted on the \emph{Aqua} satellite and the VIIRS (Visible Infrared Imaging Radiometer Suite) sensors mounted on the \emph{NOAA 20} and \emph{Suomi-NPP} satellites. Image credit: NASA Ocean Color Web (\href{https://oceancolor.gsfc.nasa.gov/gallery/482/} {https://oceancolor.gsfc.nasa.gov/gallery/482/}).} 
  \label{fig:kelvin}
\end{figure}

In \cite{Beron-21-POFb}, the IL$^0$PE was extended to include \emph{stratification} in the form of a polynomial of arbitrary degree $\alpha$ in the vertical coordinate, while maintaining its two-dimensional nature. This extended model is referred to as IL$^{(0,\alpha)}$PE, where the first slot indicates that the velocity does not explicitly vary vertically and the second slot denotes the degree of vertical variation allowed for the density. (The IL$^{(0,0)}$PE corresponds to the IL$^0$PE, and the IL$^{(0,1)}$PE appeared in a three-layer model for equatorial dynamics developed in \cite{Schopf-Cane-83}.) The IL$^{(0,\alpha)}$PE enhances the physics of the IL$^0$PE by facilitating the representation of additional processes, notably mixed-layer restratification by baroclinic instability \cite{Boccaletti-etal-07}.

More vertical variation might be added while maintaining the two-dimensional structure of the HLPE, which is essential for facilitating basic physical understanding. This has been the main motivation for developing the IL$^0$PE. For instance, an IL$^{(\hat\alpha,\alpha)}$PE would include vertical velocity shear in the form of a polynomial up to degree $\hat\alpha$. In \cite{Ripa-JFM-95}, an IL$^{(1,1)}$PE or IL$^1$PE was developed in an attempt to better represent the dynamics of the fully three-dimensional PE, which in the notation above would be the IL$^\infty$PE. Nonetheless, the IL$^1$PE, or more broadly the IL$^{(\hat\alpha,\alpha)}$PE, still awaits proof of exhibiting the geometric structure that this paper aims to investigate.

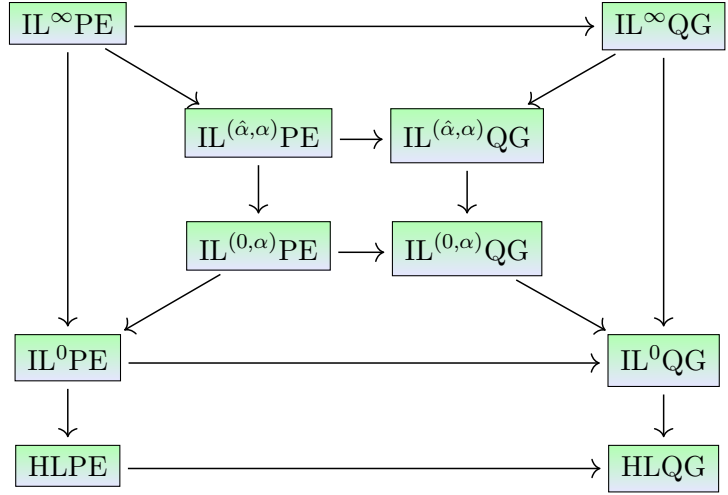
\begin{figure}[h!]
   \centering
   \begin{tikzcd}
      [row sep = 2em, 
      column sep = 2em, 
      cells = {nodes={top color=green!30, bottom color=blue!10,draw=black}},
      arrows = {draw = black, rightarrow, line width = .02cm, , shorten <= 1mm, shorten >= 1mm}]
      \mathrm{IL}^\infty \mathrm{PE} \arrow[ddd] \arrow[dr] \arrow[rrr] & & & \mathrm{IL}^\infty\mathrm{QG} \arrow[dl] \arrow[ddd] \\
      & \mathrm{IL}^{(\hat{\alpha},\alpha)}\mathrm{PE} \arrow[d] \arrow[r] & \mathrm{IL}^{(\hat{\alpha},\alpha)}\mathrm{QG} \arrow[d] & \\
      & \mathrm{IL}^{(0,\alpha)}\mathrm{PE} \arrow[dl] \arrow[r] & \mathrm{IL}^{(0,\alpha)}\mathrm{QG} \arrow[dr] & \\
      \mathrm{IL}^0 \mathrm{PE} \arrow[d] \arrow[rrr] & & & \mathrm{IL}^0\mathrm{QG} \arrow[d] \\
      \mathrm{HLPE} \arrow[rrr] & & & \mathrm{HLQG}
   \end{tikzcd}
   \caption{Overview of the models discussed in this work. Each arrow indicates an approximation. The most general model is the IL$^\infty$PE model, which represents the fully three-dimensional primitive equations, where IL stands for inhomogeneous layer. The IL$^\infty$QG model is the quasi-geostrophic limit of IL$^\infty$PE and is also a fully three-dimensional model. The interior of the diagram features models that have polynomial approximations of vertical shear and stratification. The order of these polynomials is represented by $\hat{\alpha}$ for vertical shear and by $\alpha$ for stratification. IL$^0$ indicates that the model has no vertical shear or stratification, but does include horizontal variations of buoyancy. The homogeneous layer models indicated by HL have no horizontal variations of buoyancy.}
  \label{fig:introductiontree}
\end{figure}

\paragraph{Geometric structure.}

The IL$^{(0,\alpha)}$PE was shown in \cite{Beron-21-POFb} to admit \emph{Euler--Poincar\'e variational formulation} and possess \emph{Lie--Poisson Hamiltonian structure}. By Euler--Poincar\'e variational principle, we refer to a Hamilton's principle for fluids that leads to the motion equations in Eulerian variables. This is achieved by expressing the Lagrangian in terms of Eulerian variables and then extremizing the corresponding action under constrained variations, which represent fluid particle path variations at fixed Lagrangian labels and time. This principle, as described, was apparently first discussed in \cite{Newcomb-62}. However, it is part of a much broader variational formulation of mechanics, written in the abstract language of differential geometry by \cite{Holm-etal-98a, Holm-etal-02}, who made connections with the seminal work of Henri Poincar\'e \cite{Poincare-10}, built on earlier work by Hamilton and Lie.

By Lie--Poisson Hamiltonian structure, we refer to a type of noncanonical Hamiltonian representation of the equations of motion in terms of Eulerian variables, as discussed by \cite{Morrison-Greene-80}. The abstract formulation, rooted in previous work by \cite{Arnold-66b}, is due to \cite{Marsden-Weinstein-82, Marsden-Weinstein-83}. In particular, \cite{Marsden-etal-84} demonstrated how to derive the Euler equation for compressible fluid motion as a Lie--Poisson Hamiltonian system through reduction by symmetry of the corresponding canonical Hamiltonian formulation, i.e., the Euler equation written in Lagrangian variables.

\paragraph{Goal of the paper and organization.}

The Lie--Poisson Hamilton equations are connected to the Euler--Poincar\'e equations via a (generally partial) \emph{Legendre transform}. This connection has been elusive for the quasigeostrophic (QG) approximation to the PE, which represents a sub-Coriolis or inertial frequency approximation to the PE. Recently, \cite{Luesink-etal-24} clarified this link for the HLQG. Building on their work, we develop an Euler--Poincar\'e variational principle for the IL$^{(0,\alpha)}$QG in this paper and derive the IL$^{(0,\alpha)}$QG in Lie--Poisson Hamiltonian form through a Legendre transform.  Previously, this derivation was performed by direct manipulation in \cite{Beron-21-POFb}. An important byproduct of the Euler--Poincar\'e variational formalism is the obtention of a \emph{Kelvin--Noether theorem} for the circulation of an appropriately modified velocity along a material loop, generalizing the one derived in \cite{Beron-Olascoaga-24} for the IL$^{(0,1)}$QG by direct manipulation.  A key feature of the Lie--Poisson formulation is the presence of conservation laws, originally identified by Sophus Lie as ``distinguished functionals'' \cite{Lie-80} and now commonly referred to as \emph{Casimirs} \cite{Marsden-Ratiu-99} following nomenclature introduced, apparently, in \cite{Sudarshan-Makunda-74}, which commute with any function within the Lie--Poisson bracket.  These conservation laws are not tied to explicit symmetries like energy or momentum but are related to symmetries under the relabeling of fluid particles, which are hidden within the Eulerian variables. We identify the Casimirs of the system and we explicitly recognize them as Noether quantities arising from the Euler--Poincar\'e variational principle.  (The connection between material conservation of vorticity or potential vorticity in geophysical context and particle relabeling symmetry has a long history \cite{Newcomb-67, Ripa-AIP-81, Salmon-82, Henyey-82, Padhye-Morrison-96a}. Our analysis is motivated by \cite{Cotter-Holm-12}, who explored boundary terms in the Euler--Poincar\'e variational principle but did not establish a link to these Noether quantities as we do here.)  The approach taken in this paper is mainly algebraic, with the differential geometry components  restricted to a section that can be skipped without breaking the flow, unless the reader wishes to explore the geometric nature of the Legendre transform and its related \emph{momentum map}, as well as understand the origins of the term ``Kelvin--Noether.''  

The remainder of the paper is organized as follows. In Sec.\@~\ref{sec:ILQG}, a review of the IL$^{(0,\alpha)}$QG is presented. In Sec.\@~\ref{sec:prep} we outline a few basic assumptions and provide a definition that will allow us to formally and selfconsistently build a dual Euler--Poincaré/Lie--Poisson formulation of the IL$^{(0,\alpha)}$QG. Section \ref{sec:EP} is devoted to the elaboration of the Euler--Poincar\'e variational principle for the IL$^{(0,\alpha)}$QG, culminating in the establishment of a Kelvin--Noether circulation theorem. The discussion in Sec.\@~\ref{sec:LP} focuses on the Legendre transform, which leads to the derivation of IL$^{(0,\alpha)}$QG in the Lie-Poisson Hamiltonian form. Nother's theorem for generalized Hamiltonian systems is treated in Sec.\@~\ref{sec:conslaws} in connection of energy conservation and the emergence of Casimir invariants. Particle relabeling symmetry and conservation of Casimirs are treated in Sec.\@~\ref{sec:Noether}. In Sec.\@~\ref{sec:geomech}, a geometric mechanics interpretation of the findings from the previous sections is provided, offering generalization. The paper ends with a recap and recommendations for future research in Sec.\@~\ref{sec:concl}.

\section{The IL$^{(0,\alpha)}$QG}\label{sec:ILQG}

Let $\mathbf x = (x,y) \in \mathbb R^2$ denote the position in a domain $D$ of the $\beta$-plane with external unit normal $\hat{\mathbf n}$ to its boundary, $\partial D$. (More complex geometrical configurations, such as those with multiple connections, can be handled with minimal additional effort.)  Assume a reduced-gravity setting. Let
\begin{equation}
  R := \frac{\sqrt{g'H}}{|f_0|},
  \label{eq:Rd}
\end{equation}
where $H$ represents the thickness of the active layer in a \emph{reference state} with no motion, $f_0$ stands for the mean Coriolis parameter, and $g' > 0$ is a parameter to be identified later. Let further
\begin{equation}
  S := \frac{N_0^2H}{2g'}
  \label{eq:S}
\end{equation}
such that $0 < S < 1$, where $N_0 > 0$ is another parameter to be identified.

\begin{center}
\fbox{%
\begin{minipage}{\textwidth}
The IL$^{(0,\alpha)}$QG, $\alpha \in \mathbb Z_0^+$, in the above setting is given by \cite{Beron-21-POFb}
\begin{equation}
  \partial_t \bar\xi + [\bar\psi,\bar\xi] = R^{-2}_\alpha[\bar\psi,\nu(\psi_\sigma, \psi_{\sigma^2}, \dotsc,\psi_{\sigma^{\alpha+1}})],\quad \partial_t \psi_{\sigma^n} + [\bar\psi,\psi_{\sigma^n}] = 0,
\label{eq:ILQG}
\end{equation}
$n=1,2,\dotsc,\alpha+1$, where 
\begin{equation}
  \nu(\psi_\sigma, \psi_{\sigma^2}, \dotsc, \psi_{\sigma^{\alpha+1}}) := \psi_\sigma - \sum_{n=1}^\alpha(n+1)\overline{\sigma^{n+1}}\psi_{\sigma^{n+1}}
  \label{eq:nu}
\end{equation}
with 
\begin{equation}
  \bar\psi = (\nabla^2-R^{-2}_\alpha)^{-1}(\bar\xi - R^{-2}_\alpha \nu(\psi_\sigma, \psi_{\sigma^2}, \dotsc, \psi_{\sigma^{\alpha+1}}) - \beta y), 
  \label{eq:inv}
\end{equation}
which is subject to
\begin{equation}
  \nabla^\perp\bar\psi\cdot\hat{\mathbf n}\vert_{\partial D} = 0,\quad \frac{d}{dt}\oint_{\partial D} \nabla^\perp\bar\psi\cdot d\mathbf x = 0.
  \label{eq:bc}
\end{equation}
\end{minipage}
}
\end{center}

\noindent Here, 
\begin{equation}
  [a,b] := \hat{\mathbf z}\cdot\nabla a\times\nabla b =: \nabla^\perp a\cdot \nabla b = - \partial_ya\partial_xb + \partial_xa\partial_yb,
\end{equation}
where $\hat{\mathbf z}$ is the vertical unit vector, is the Jacobian of the map $\mathbf x \mapsto (a(\mathbf x),b(\mathbf x))$.  The parameter
\begin{equation}
  R^2_\alpha := \Big(1 - \tfrac{1}{2}\sum_{n=1}^\alpha \overline{\sigma^{n+1}} S\Big)R^2,
\end{equation}
where
\begin{equation}
  \sigma := 1 + 2\frac{z}{H}
\end{equation}
is a rescaled vertical coordinate that varies (linearly) from $+1$ at the surface ($z = 0$), where a rigid lid is placed, down to $-1$ at the bottom of the active layer, which in the QG limit, clarified below, effectively coincides with that of the reference state, lying at $z = - H$.  The overbar denotes a vertical average across this range.  Finally, the inverse of
\begin{equation}
  \nabla^2-R^{-2}_\alpha = \partial_{xx} + \partial_{yy}- R^{-2}_\alpha 
\end{equation}
is interpreted in terms of the relevant Green function for the elliptic problem
\eqref{eq:inv}--\eqref{eq:bc}.

The IL$^{(0,\alpha)}$QG thus has $\alpha+2$ prognostic fields, given by $(\bar\xi, \psi_\sigma, \psi_{\sigma^2}, \dotsc, \psi_{\sigma^{\alpha+1}})$, which are assumed to be smooth in each of its arguments, $(\mathbf x,t)$.  These diagnose $\bar\psi(\mathbf x,t)$ via \eqref{eq:inv}, which defines the invertibility principle for the IL$^{(0,\alpha)}$QG.

\subsection{Physical interpretation of the model fields}

The velocity in the IL$^{(0,\alpha)}$PE is horizontal and vertically shearless.  We write this field as ${\overline{\mathbf u}}^h(\mathbf x,t)$, representing a vertically averaged field from $z = 0$ down to $z = -h(\mathbf x,t)$, the soft interface with the inert abyssal layer.  The buoyancy,
\begin{equation}
  \vartheta(\mathbf x,z,t) := -g\frac{\rho(\mathbf x,z,t) - \rho_\mathrm{inert}}{\rho_0},
\end{equation}
where $g$ denotes the acceleration due to gravity, $\rho$ represents the density of the active layer, $\rho_\mathrm{inert} = \mathrm{const}$ is the density of the inert layer, and $\rho_0$ is the density used in the Boussinesq approximation. This is written in the IL$^{(0,\alpha)}$PE as
\begin{equation}
  \vartheta(\mathbf x,z,t) = {\overline{\vartheta}}^h(\mathbf x,t) + \sum_{n=1}^\alpha (\sigma_h^n - {\overline{\sigma_h^n}}^h) \vartheta_{\sigma^n}(\mathbf x,t) 
  \label{eq:buoyancy}
\end{equation}
where
\begin{equation}
  \sigma_h := 1 + 2\frac{z}{h}.
\end{equation}
The coefficients of this expansion are materially conserved by the flow. 

Let $\mathrm{Ro} > 0$ be a small parameter taken to represent a Rossby number, measuring the strength of inertial and Coriolis forces, e.g.,
\begin{equation}
  \mathrm{Ro} = \frac{V}{|f_0|R} \ll 1, 
  \label{eq:Ro}
\end{equation}
where $V$ is a characteristic velocity. The QG scaling \cite{Pedlosky-87}  asserts that
\begin{equation}
  (|{\overline{\mathbf u}}^h|,h-H,\partial_t,\beta y) = O(\mathrm{Ro}\hspace{0.08334em} V, \mathrm{Ro}\hspace{0.08334em} R, \mathrm{Ro} f_0,  \mathrm{Ro} f_0).  
  \label{eq:QGscaling}
\end{equation}
Consistent with this scaling, with an $O(\mathrm{Ro}^2)$ error, we have that
\begin{align} 
  {\overline{\mathbf u}}^h &=  \nabla^\perp\bar\psi,\label{eq:u-IL}\\ h &= H + \frac{H}{f_0R_S^2}\left(\bar\psi - \psi_\sigma + \sum_{n=1}^\alpha\overline{\sigma^{n+1}}\psi_{\sigma^{n+1}}\right) \equiv H + \frac{H}{f_0R_S^2}(\bar\psi - \nu)\label{eq:h}\\ \bar\vartheta &= g' + \frac{2g'}{f_0R^2}\psi_\sigma,\\ \vartheta_\sigma &= \tfrac{1}{2}N_0^2H + \frac{4g'}{f_0R^2}\psi_{\sigma^2},\\ \vartheta_{\sigma^n} &= \frac{2(n+1)g'}{f_0R^2}\psi_{\sigma^{n+1}},
\end{align}
$n = 2,3,\dotsc,\alpha$. Finally, with an $O(\mathrm{Ro}^2)$ error, the potential vorticity in the IL$^{(0,\alpha)}$PE
\begin{equation}
  \frac{\nabla^\perp\cdot{\overline{\mathbf u}}^h + f}{h} = \frac{f_0 + \bar\xi}{H}.  
  \label{eq:q-IL}
\end{equation}

With the identifications \eqref{eq:u-IL}--\eqref{eq:q-IL}, the following interpretations apply. 
\begin{enumerate}
  \item Parameter \eqref{eq:S} measures stratification in the reference state characterized by $\bar\xi = \beta y$ and $\psi_{\sigma^n} = 0$, $n = 1,2,\dotsc,\alpha+1$, implying $\bar\psi = 0$, i.e., no motion.  Indeed, in that state, the stratification is uniform, with the buoyancy varying from $g'(1 - S)$ at the bottom of the layer to $g'(1 + S)$ at the surface.  Thus $g'(1 - S)$ represents the reference buoyancy at the base of the layer, with $g'$ representing the reduced gravity in the absence of reference stratification ($S=0$). Parameter $N_0$ in \eqref{eq:S} is the reference Brunt--V\"ais\"al\"a frequency as its squared is equal to the $z$-derivative of the reference buoyancy.

  \item Parameter \eqref{eq:Rd} is thus interpreted as the equivalent-barotropic Rossby radius of deformation of the system, approximately representing the gravest-baroclinic deformation radius in a model extending from the ocean surface down to the ocean floor.

  \item The equation on the left of \eqref{eq:ILQG} controls the evolution of IL$^{(0,\alpha)}$QG potential vorticity. This quantity is not materially conserved.  Note that the material derivative
  \begin{equation}
    \frac{D}{Dt} = \partial_t + \nabla^\perp\bar\psi\cdot\nabla = \partial_t + [\bar\psi,\,\,].
  \end{equation}
  Thus $\bar\xi$ is created (or annihilated) by the misalignment between the gradients of buoyancy and layer thickness. This is consistent with the lack of material conservation of Ertel's $\frac{z}{h}$-potential vorticity in the IL$^\infty$PE, ${\overline{q}}^h$, as obtained when the horizontal velocity in that model is replaced by ${\overline{\mathbf u}}^h$; the IL$^{(0,\alpha)}$PE potential vorticity is proportional to ${\overline{q}}^h$ \cite{Ripa-JFM-95}.

  \item The remaining equations in \eqref{eq:ILQG} are statements of material conservation of the vertical average, vertical derivative, etc., of the buoyancy.

  \item The boundary conditions \eqref{eq:bc} represent zero-flow across $\partial D$ and constancy of Kelvin circulation along $\partial D$.
\end{enumerate}

Finally, by the thermal-wind balance, the velocity in the IL$^{(0,\alpha)}$QG has implicit vertical shear, which motivates the streamfunction notations for the buoyancy \cite{Beron-21-POFb}. Specifically, the buoyancy distribution \eqref{eq:buoyancy} implicitly implies that the velocity is determined, with an $O(\mathrm{Ro}^2)$ error, by the streamfunction
\begin{equation}
  \psi = \bar\psi +  \sum_{n=1}^{\alpha+1}(\sigma^n - \overline{\sigma^n})\psi_{\sigma^n}.
\end{equation}

\subsection{The IL$^0$QG as a special case of the IL$^{(0,\alpha)}$QG}

Making $\alpha = 0$, which means ignoring the terms $\psi_{\sigma^n}$, $n = 2,3,\dotsc,\alpha+1$, and setting $S = 0$, the IL$^{(0,\alpha)}$QG reduces to the IL$^0$QG.  Explicitly, the IL$^0$QG reads
\begin{equation}
  \partial_t\bar\xi + [\bar\psi,\bar\xi] = R^{-2}[\bar\psi,\psi_\sigma],\quad \partial_t\psi_\sigma + [\bar\psi,\psi_\sigma] = 0,
\end{equation}
where
\begin{equation}
  \bar\psi = (\nabla^2 - R^{-2})^{-1}(\bar\xi - R^{-2}\psi_\sigma - \beta y).
\end{equation}
The IL$^0$QG as above appears in \cite{Ripa-DAO-99, Beron-21-POFa} and in nondimensional form in \cite{Warneford-Dellar-14}.

\begin{remark}\label{rem:tqg}
  The ``TQG'' discussed in \textup{\cite{Holm-etal-21}} reads, in dimensional variables and ignoring topographic forcing, as
  \begin{equation}
	\partial_t(\bar q + R^{-2}\psi_\sigma) + [\bar\psi,\bar q] = 0,\quad \partial_t\psi_\sigma + [\bar\psi, \psi_\sigma] = 0,
  \end{equation}
  where
  \begin{equation}
	 \bar q := \bar\xi - R^{-2}\psi_\sigma.
  \end{equation}
  The TQG thus is identical to the IL$^0$QG.  
\end{remark}

\begin{remark}\label{rem:ilqgm}
  The ``ILQGM'' discussed in \textup{\cite{Ripa-RMF-96}} is a quasigeostrophic approximation to the IL$^0$PE that differs from that one that leads the IL$^0$QG in that it considers the most general motionless reference state in the IL$^0$PE, which is characterized by
  \begin{equation}
	 \bar\Theta(\mathbf x)/\kappa(\mathbf x)^2 = g'.
	 \label{eq:kappa}
  \end{equation}
  Here, $\bar\Theta(\mathbf x)$ is the buoyancy in the reference state and $H/\kappa(\mathbf x)$ gives the layer thickness in that state.  With this in mind, the ILQGM reads
  \begin{equation}
	\partial_t\bar\xi + \kappa[\bar\psi,\bar\xi] = \kappa R^{-2}\left[\bar\psi, \kappa\psi_\sigma\right],\quad \partial_t\psi_\sigma + \left[\bar\psi, \kappa\psi_\sigma\right] = 0, 
	\label{eq:ILQGM}
  \end{equation}
  where
  \begin{equation}
	 \bar\psi = (\nabla\cdot\kappa\nabla - R^{-2})^{-1}(\kappa^{-1}\bar\xi - R^{-2}\psi_\sigma - \beta y).
  \end{equation}
  The IL$^0$QG follows the ILQGM upon setting $\kappa = 1$.
\end{remark}

\subsection{Invariant sub-dynamics of IL$^{(0,\alpha)}$QG dynamics}

If the IL$^{(0,\alpha)}$QG is initialized from $\psi_{\sigma^n} = \const$, $n = 1, 2, \dotsc, \alpha+1$, then these variables preserve their initial constant values all the time.  In other words, the subspace $\{\psi_\sigma, \psi_{\sigma^2}, \dotsc, \psi_{\sigma^{\alpha+1}} = \const\}$ represents an invariant subspace of the IL$^{(0,\alpha)}$QG. The dynamics on this invariant subspace is formally the same as that of the HLQG, in which case the potential vorticity, given by $\bar\xi = \nabla^2\bar\psi - R^{-2}\bar\psi + \beta y$, is materially conserved.  This holds formally because $\psi_{\sigma^n}$, $n = 1, 2, \dotsc, \alpha+1$, represent perturbations on a reference uniform stratification.  This is reflected in the IL$^{(0,\alpha)}$QG through the stratification parameter $S$. The HLQG and IL$^{(0,\alpha)}$QG potential vorticities on $\{\psi_\sigma, \psi_{\sigma^2}, \dotsc, \psi_{\sigma^{\alpha+1}} = \const\}$ differ, except for unimportant constants, by $R_\alpha$ being smaller than $R$ being smaller for $S>0$.

If the IL$^{(0,\alpha)}$QG is initialized from $\psi_{\sigma^n} = \const$, $n = 2, \dotsc, \alpha+1$, then these quantities are preserved for all time.  The dynamics on this invariant subspace is formally the same as that of the IL$^0$QG, with the caveats noted above.

In particular, if the IL$^0$QG is initialized with $\psi_\sigma = \mathrm{const}$, this is preserved for all time by material conservation of $\psi_\sigma$. The dynamics on the $\{\psi_\sigma = \mathrm{const}\}$ subspace coincide with that of the HLQG, exactly.

%

\section{Preparation}\label{sec:prep}

A few considerations are required in order to establiblish the dual Euler--Poincar\'e/Liee-Poisson formulation of the IL$^{(0,\alpha)}$QG.

\begin{assumption}\label{ass:cir-cons}
  Let 
  \begin{equation}
	\gamma := \oint_{\partial D} \nabla^\perp\bar\psi\cdot d\mathbf x = 0
  \end{equation}
  be the circulation of the velocity along the boundary of the flow domain, $\partial D$.  We will assume that $\gamma$ is constant, namely,
  \begin{equation}
      \dot\gamma = 0.
  \end{equation}
\end{assumption}

The imposition of the condition $\dot\gamma = 0$ is essential to ensure that the IL$^{(0,\alpha)}$QG, in the absence of external forcing and dissipative effects, conserves energy, as explictly shown below. An exception to this arises when $R_\alpha\to\infty$, implying that the lower boundary of the active fluid layer behaves as a rigid interface. In such a case, $\dot\gamma = 0$ is inherently satisfied by the system's dynamics. This is well-documented within the framework of QG theory \cite{Pedlosky-87}.

To further guarantee volume preservation, an additional assumption is needed.

\begin{assumption}\label{ass:vol}
  We will assume that
  \begin{equation}
	\frac{d}{dt}\int_D \bar\psi\,d^2x = 0.
  \end{equation}
\end{assumption}

A discussion on this requirement in the HLQG may be found in \cite{Gent-McWilliams-83}.  We will make explicit the implication for volume conservation in the IL$^{(0,\alpha)}$QG in the subsequent section.

The variational principles used in this paper are based on the concept of the variational derivative, which is first explained.

\begin{definition}\label{def:fun}
  Let $\mathcal{B}$ be a Banach space. If $\mathscr F:\mathcal{B}\to\mathbb{R}$ is a functional, denote by $\langle\,,\hspace{.05cm}\rangle$ the pairing on $\mathcal B $. Let $u,v \in \mathcal B$, then we define the \textbf{first variation of $\mathscr F\!\!\!\!\!\!\mathscr F$} using the Gateaux derivative as 
  \begin{equation}
	   \delta \mathscr F(u) := \left.\frac{d}{d\epsilon}\right\vert_{\epsilon = 0} \mathscr F(u + \epsilon v) = \left\langle\frac{\delta\mathscr F}{\delta u},v\right\rangle.
  \end{equation}
  This defines the \textbf{functional derivative of $\mathscr F\!\!\!\!\!\!\mathscr F$} uniquely as $\frac{\delta}{\delta u}\mathscr F \in \mathcal B^*$ since $v$ can be arbitrary. This function is called the \textbf{variation of $\boldsymbol u$} and is denoted as $\delta u$.  For more details, see \textup{\cite{Gelfand-Fomin-00}}. 
\end{definition}

The relevant pairing $\langle\,,\hspace{.05cm}\rangle : \mathcal B^* \times \mathcal B \to \mathbb R$ in our context is the $L^2$-pairing.    

\begin{assumption}\label{ass:cir-var}
  We will assume that the velocity circulation along $\partial D$ vanishes identically:
  \begin{equation}
      \delta\gamma = 0.
  \end{equation}
\end{assumption}

The restriction $\delta\gamma = 0$ allows for a variational formulation of the IL$^{(0, \alpha)}$QG using variational calculus consistent with Def.\@~\ref{def:fun}. Otherwise, the phase-space variables must be expanded (in the Hamiltonian formulation) to include the velocity circulation, as, for instance, done in \cite{Holm-etal-85}. However, this approach would require redefining the notion of a variational derivative differently from Def.\@~\ref{def:fun}. A proposal for this redefinition is presented in \cite{Lewis-etal-86}.

Finally, the variables in the IL$^{(0,\alpha)}$QG are (smooth, time-dependent) \emph{scalar-valued functions} on $D$.  \emph{We denote the space of such fields by $\mathcal F(D)$.}  

\section{Euler-Poincar\'e variational principle for the IL$^{(0,\alpha)}$QG}\label{sec:EP}

With the various considerations in the preceding section, we are ready to announce and prove our first theorem.

\begin{theorem}[Euler--Poincar\'e for IL$^{(0,\alpha)}$QG]\label{thm:EPforILQG}
  The IL$^{(0,\alpha)}$QG follows from \textbf{Euler--Poincar\'e' variational principle}, that is, a Hamilton's principle
  \begin{equation}
	 \delta\int_{t_0}^{t_1}\mathscr L(\bar\psi,\psi_\sigma,\psi_{\sigma^2},\dotsc,\psi_{\sigma^{\alpha+1}})\,dt = 0
    \label{eq:dS}
  \end{equation}
  constrained to
  \begin{equation}
	 \delta\bar\psi = \partial_t\eta + [\bar\psi,\eta],\quad
	 \delta\psi_{\sigma^n} = - [\eta,\psi_{\sigma^n}], 
	 \label{eq:cons}
  \end{equation}
  $n = 1,2,\dotsc,\alpha+1$, where $\eta$ vanishes at the endpoints of integration and is otherwise arbitrary, with \textbf{Lagrangian} defined by
  \begin{equation}
	 \mathscr L(\bar\psi,\psi_\sigma,\dotsc,\psi_{\sigma^{\alpha+1}}) := \frac{1}{2}\int_D |\nabla\bar\psi|^2 + R^{-2}_\alpha\bar\psi^2 - 2\beta y\bar\psi + 2R^{-2}_\alpha \nu(\psi_\sigma,\dotsc,\psi_{\sigma^{\alpha+1}})\bar\psi\,d^2x.
    \label{eq:L}
  \end{equation}
\end{theorem}

\begin{proof}
  The proof begins with the computation of the variation of the action
  $\int_{t_0}^{t_1}\mathscr L\,dt$:
  \begin{align}
	 \delta\int_{t_0}^{t_1}\mathscr L\,dt &= \int_{t_0}^{t_1}\left\langle\frac{\delta \mathscr L}{\delta\bar\psi},\delta\bar\psi\right\rangle + \sum_{n=1}^{\alpha+1} \left\langle\frac{\delta \mathscr L}{\delta\psi_{\sigma^n}}, \delta\psi_{\sigma^n}\right\rangle\, dt \nonumber\\ &=\int_{t_0}^{t_1}\left\langle\frac{\delta \mathscr L}{\delta\bar\psi},\partial_t\eta + [\bar\psi,\eta]\right\rangle + \sum_{n=1}^{\alpha+1} \left\langle\frac{\delta \mathscr L}{\delta\psi_{\sigma^n}}, [\psi_{\sigma^n},\eta]\right\rangle\,dt\nonumber \\ &= - \int_{t_0}^{t_1}\left\langle\partial_t\frac{\delta \mathscr L}{\delta\bar\psi} + \left[\bar\psi,\frac{\delta \mathscr L}{\delta\bar\psi}\right] + \sum_{n=1}^{\alpha+1} \left[\psi_{\sigma^n},\frac{\delta \mathscr L}{\delta\psi_{\sigma^n}}\right],\eta\right\rangle\,dt,
  \end{align}
  upon integrating by parts where we have used
  \begin{align}
	 \langle a,[\bar\psi,b]\rangle = \int_D a[\bar\psi,b]\,d^2x &= \int_D a\nabla^\perp\bar\psi\cdot\nabla b\,d^2x\nonumber\\ &= \oint_D ab\nabla^\perp\psi\cdot\hat{\mathbf n}ds - \int_D b\nabla\cdot a\nabla^\perp\bar\psi\,d^2x\nonumber\\ & = - \int_D b [\bar\psi,a]\,d^2x = - \langle b,[\bar\psi,a]\rangle
  \label{eq:dual}
  \end{align}
  for every $a,b(\mathbf x)$. The penultimate equality holds by the no-flow condition through $\partial D$, given by the left equation in \eqref{eq:bc}.  The action is extremized for all $\eta$ (subject to $\eta(t_0) = \eta(t_1) = 0$) when
  \begin{equation}
	 \partial_t\frac{\delta \mathscr L}{\delta\bar\psi} + \left[\bar\psi,\frac{\delta \mathscr L}{\delta\bar\psi}\right] = \sum_{n=1}^{\alpha+1} \left[\frac{\delta \mathscr L}{\delta\psi_{\sigma^n}},\psi_{\sigma^n}\right].
	 \label{eq:EPforILQG}
  \end{equation}
  Then one computes
  \begin{align}
	 \delta\mathscr L &= \int_D\nabla\bar\psi\cdot\nabla\delta\bar\psi + R_\alpha^{-2}\bar\psi\delta\bar\psi - \beta y\delta\bar\psi + R_\alpha^{-2}\nu\delta\bar\psi + R_\alpha^{-2}\bar\psi\sum_{n=1}^{\alpha+1}\frac{\partial\nu}{\partial \psi_{\sigma^n}}\delta \psi_{\sigma^n}\,d^2x\nonumber\\ &= \int_D \left(-\nabla^2\bar\psi + R_\alpha^{-2}\bar\psi - \beta y + R_\alpha^{-2}\nu\right)\delta\bar\psi + R_\alpha^{-2}\bar\psi\left(\delta\psi_\sigma - \sum_{n=1}^\alpha (n+1)\overline{\sigma^{n+1}}\delta \psi_{\sigma^n}\right)\,d^2x\nonumber\\ &= \langle -\bar\xi,\delta\bar\psi \rangle + \langle R_\alpha^{-2}\bar\psi,\delta\psi_\sigma \rangle - \sum_{n = 1}^\alpha\left\langle R_\alpha^{-2}(n+1)\overline{\sigma^{n+1}}\bar\psi,\delta\psi_{\sigma^n} \right\rangle,
  \end{align}
  upon integrating by parts with Assump.\@~\eqref{ass:cir-var} in mind. Thus,
  \begin{equation}
	 \frac{\delta\mathscr L}{\delta\bar\psi} = -\bar\xi,\quad \frac{\delta\mathscr L}{\delta\psi_\sigma} = R^{-2}_\alpha\bar\psi,\quad \frac{\delta\mathscr L}{\delta\psi_{\sigma^{n+1}}} = - R^{-2}_\alpha(n+1)\overline{\sigma^{n+1}}\bar\psi,
    \label{eq:dLdpsi}
  \end{equation}
  $n = 1,2,\dotsc,\alpha$. 

  The proof is completed upon noting that the constraint on $\delta\psi_{\sigma^n}$ ($n=1, 2, \dotsc, \alpha+1$) in \eqref{eq:cons} is equivalent to material conservation of $\psi_{\sigma^n}(\mathbf x,t)$. To see this, let $\psi_{\sigma^n}^t(\mathbf x) := \psi_{\sigma^n}(\mathbf x,t)$.  Material conservation means
  \begin{equation}
	 \psi_{\sigma^n}^t(\mathbf x) =  \psi_{\sigma^n}^0(\mathbf q),
	 \label{eq:mat}
  \end{equation}
  where $\mathbf q$ is the position occupied by a fluid particle at time $t = 0$. The position of this particle at a latter time $t$ is $\mathbf x$.  Now, let $\eta(\mathbf x,t)$ be defined by
  \begin{equation}
	 \nabla^\perp\eta(\mathbf x(t),t) := \delta\mathbf x(t).
	 \label{eq:eta}
  \end{equation}
  Then from \eqref{eq:mat} we compute
  \begin{equation}
	 \delta\vert_{\mathbf q,t}\psi_{\sigma^n}^t + \nabla \psi_{\sigma^n}^t\cdot \delta\mathbf x = 0.
  \end{equation}
  Using \eqref{eq:eta}, the constrains on $\delta\psi_{\sigma^n}$, $n = 1, 2, \dotsc, \alpha+1$, in \eqref{eq:cons} follow, finalizing the proof.
\end{proof}

\subsection{Clarification of the constraint on $\bar\psi$}

We believe it is instructive to clarify the constraint on $\bar\psi$ in \eqref{eq:cons}. We know that fluid particle trajectories $\mathbf x(t)$ obey
\begin{equation}
  \dot{\mathbf x}(t) = \nabla^\perp\bar\psi(\mathbf x(t),t),\quad \mathbf x(0) = \mathbf q.
  \label{eq:dxdt}
\end{equation}
Using the chain rule, the variation of this equation is
\begin{equation}
  \delta \dot{\mathbf x} = \delta\vert_{\mathbf q,t} \nabla^\perp\bar\psi + (\nabla\nabla^\perp\bar\psi)\delta\mathbf x.
\end{equation}
Similarly, the time derivative of \eqref{eq:eta} reads
\begin{equation}
  \partial_t\nabla^\perp\eta + (\nabla\nabla^\perp\eta)\dot{\mathbf x} = \frac{d}{dt}\delta\mathbf x.
\end{equation}
By commutativity of differentiation, $\frac{d}{dt}\delta\mathbf x = \delta\dot{\mathbf x}$. Using \eqref{eq:eta} and \eqref{eq:dxdt}, one finally finds
\begin{equation}
  \nabla^\perp\delta\bar\psi = \nabla^\perp\partial_t\eta + (\nabla^\perp\bar\psi\cdot\nabla)\nabla^\perp\eta - (\nabla^\perp\eta\cdot\nabla)\nabla^\perp\bar\psi = \nabla^\perp(\partial_t\eta + [\bar\psi,\eta]),
\end{equation}
which leads to the left equation in \eqref{eq:cons}.

\subsection{Kelvin--Noether's circulation theorem for the IL$^{(0,\alpha)}$QG}\label{sec:kelvin}

Let $D_t$ be a material region which at time $t=0$ occupied position $D_0$.  Let $J := \partial(\mathbf x)/\partial(\mathbf q)$ be the Jacobian of the transformation $\mathbf q\mapsto \mathbf x$, assumed to be smoothly invertible, so $J > 0$.  Recall the Euler formula of fluid mechanics,
\begin{equation}
  \frac{1}{J}\frac{DJ}{Dt} = \nabla\cdot \nabla^\perp\bar\psi =
  0.
\end{equation}
Then, by changing variables we compute
\begin{equation}
  \frac{d}{dt}\int_{D_t} \frac{\delta\mathscr L}{\delta\bar\psi}\,d^2x = \int_{D_0} \frac{D}{Dt}\left(\frac{\delta\mathscr L}{\delta\bar\psi}J\right)\,d^2q = \int_{D_t} \partial_t\frac{\delta \mathscr L}{\delta\bar\psi} + \left[\bar\psi,\frac{\delta \mathscr L}{\delta\bar\psi}\right]\,d^2x.
  \label{eq:change}
\end{equation}
Now, by \eqref{eq:EPforILQG} it follows that
\begin{equation}
  \frac{d}{dt}\int_{D_t}\frac{\delta\mathscr L}{\delta\bar\psi}\,d^2x = - \sum_{n=1}^{\alpha+1} \int_{D_t}\left[\psi_{\sigma^n},\frac{\delta \mathscr L}{\delta\psi_{\sigma^n}}\right]\,d^2x,
  \label{eq:kelnoe}
\end{equation}
which is \emph{Kelvin--Noether's circulation theorem for the IL$^{(0,\alpha)}$QG}. Explicitly,
\begin{equation}
  \frac{d}{dt}\int_{D_t}\bar\xi\,d^2x = R_\alpha^{-2}\int_{D_t}[\nu,\bar\psi]\,d^2x,
  \label{eq:kel}
\end{equation}
where the definition of $\nu$, given \eqref{eq:nu}, was used.  

Let $\mathbf f(\mathbf x)$ be such that $\nabla^\perp\cdot \mathbf f = \beta y$.  Then we have, by the invertibility principle \eqref{eq:inv}, that
\begin{equation}
  \mathcal K(D_t) := \int_{D_t}\bar\xi\,d^2x = \oint_{\partial D_t}\big(\nabla^\perp\bar\psi + \mathbf f - R_\alpha^{-2}\nabla^{-2} \nabla^\perp(\bar\psi - \nu)\big)\cdot d\mathbf x,
\end{equation}
which is an appropriate definition of the Kelvin circulation along a material loop.  In general, the preservation of this flow property is not ensured, as the right-hand side of \eqref{eq:kel} typically does not vanish due to the misalignment between the gradients of buoyancy and layer thickness. This misalignment generally results in the generation (or destruction) of circulation.  Exceptions arise when $\partial D_t$ is chosen to be isopycnic and when $\partial D_t$ is taken to be the solid boundary $\partial D$ of the flow domain. 

\begin{remark}\label{rem:vol}
    Note that when $D_t$ is taken to be $D$, one has
    \begin{equation}
       \mathcal K(D) = \gamma + \int_D \beta y - R_\alpha^{-2}(\bar\psi - \nu)\,d^2x
       \label{eq:KD}
    \end{equation}
    is a constant by the Kelvin--Noether theorem.  In fact, this motion integral represents a Casimir of the IL$^{(0,\alpha)}$QG,
    discussed below.
    
    By the Kelvin--Noether theorem, Assump.\@~\eqref{ass:cir-cons} on the preservation of the velocity circulation along $\partial D$, and Assump.\@~\eqref{ass:vol} on the conservation of the integral of the integral of $\bar\psi$ over $D$, it follows that
    \begin{equation}
        \int_D \nu\,d^2x = \mathrm{const}.
    \end{equation}
    This integral of motion, as will be seen below, also represents a Casimir of the IL$^{(0,\alpha)}$QG. Finally, using Assump.\@~\eqref{ass:vol} we have
    \begin{equation}
        \mathcal V := \int_D \bar\psi - \nu\,d^2x
    \end{equation}
    is a constant whose physical interpretation is that of volume preservation, giving sustain to the need of Assumps.\@~\eqref{ass:cir-cons} and \eqref{ass:vol}.  Indeed, according to \eqref{eq:h}, $\partial_t(\bar\psi - \nu) = 0$ is nothing but the lowest-order contribution in the Rossby number \eqref{eq:Ro} of the local law of volume conservation in the IL$^{(0,\alpha)}$PE, which is given by $\partial_t h + \nabla\cdot h\bar{\mathbf u}^h = 0$.  
    
  \end{remark}

Finally, to understand why \eqref{eq:kelnoe} is denoted as a Kelvin--Noether theorem, it is necessary to revisit certain concepts of geometric mechanics.  This is explored in Sec.\@~\ref{sec:geomech} and Sec.\@~\ref{sec:mom}, specifically.

\section{Legendre transform for the IL$^{(0,\alpha)}$QG}\label{sec:LP}

We are now ready to formulate our final theorem. 

\begin{theorem}\label{thm:LTforILQG}
  The IL$^{(0,\alpha)}$QG represents a \textbf{Lie--Poisson system}. Namely, for any functional $\mathscr F(\bar\xi,\psi_\sigma,\dotsc,\psi_{\sigma^{\alpha+1}})$,
  \begin{equation}
	 \dot{\mathscr F} = \{\mathscr F,\mathscr H\} := \left\langle\bar\xi,\left[\frac{\delta \mathscr F}{\delta\bar\xi},\frac{\delta \mathscr H}{\delta\bar\xi}\right]\right\rangle + \sum_{n=1}^{\alpha+1} \left\langle\psi_{\sigma^n}, \left[\frac{\delta \mathscr F}{\delta\bar\xi},\frac{\delta \mathscr H}{\delta\psi_{\sigma^n}}\right] + \left[\frac{\delta \mathscr F}{\delta\psi_{\sigma^n}},\frac{\delta \mathscr H}{\delta\bar\xi}\right] \right\rangle
	 \label{eq:LP}
  \end{equation}
  with \textbf{Hamiltonian} given by
  \begin{equation}
	 \mathscr H(\bar\xi,\psi_\sigma,\dotsc,\psi_{\sigma^{\alpha+1}}) := \langle-\bar\xi,\bar\psi\rangle - \mathscr L(\bar\psi,\psi_\sigma,\dotsc,\psi_{\sigma^{\alpha+1}}),
	 \label{eq:H}
  \end{equation}
  subject to the \textbf{admissibility conditions}
  \begin{equation}
	 \nabla^\perp\frac{\delta F}{\delta\bar\xi}\cdot\hat{\mathbf n}\vert_{\partial D} = 0,\quad \nabla^\perp\frac{\delta F}{\delta\psi_{\sigma^n}}\cdot\hat{\mathbf n}\vert_{\partial D} = 0,.
	 \label{eq:adm-IL}
  \end{equation}
  $n = 1, 2, \dotsc, \alpha+1$.
\end{theorem}

\begin{proof}
  We begin by noting that \eqref{eq:H} implies 
  \begin{align}
	 \mathscr H &= \int_D - \left(\nabla^2\bar\psi - R^{-2}_\alpha(\bar\psi - \nu) - \beta y\right)\bar\psi - \tfrac{1}{2} \left(|\nabla\bar\psi|^2 + R^{-2}_\alpha\bar\psi^2\right) - \beta y\bar\psi - R^{-2}_\alpha \nu\bar\psi \,d^2x,\nonumber\\ &= \int_D - \left(\nabla^2\bar\psi - R^{-2}_\alpha\bar\psi \right)\bar\psi - \tfrac{1}{2} \left(|\nabla\bar\psi|^2 + R^{-2}_\alpha\bar\psi^2\right) \,d^2x,\nonumber\\ &= - \bar\psi\vert_{\partial D} \oint_{\partial D} \nabla^\perp\bar\psi\cdot d\mathbf x + \int_D |\nabla\bar\psi|^2 + R^{-2}_\alpha\bar\psi^2 - \tfrac{1}{2}\big(|\nabla\bar\psi|^2 + R^{-2}_\alpha\bar\psi^2\big)\,d^2x\nonumber\\ &= \frac{1}{2}\int_D |\nabla\bar\psi|^2 + R^{-2}_\alpha\bar\psi^2\,d^2x + \operatorname{const},
    \label{eq:H-exp}
  \end{align}
  upon integration by parts with the no-flow condition through $\partial D$, left equation in \eqref{eq:bc}, in mind.  This represents, modulo an irrelevant constant, the energy of the IL$^{(0,\alpha)}$QG. Then we compute
  \begin{align}
	 \delta \mathscr H &= \int_D \nabla\bar\psi\cdot\nabla\delta\bar\psi + R^{-2}_\alpha\bar\psi\delta\bar\psi\nonumber\\ &= \bar\psi\vert_{\partial D} \oint \nabla^\perp\delta\bar\psi\cdot d\mathbf x + \int_D \bar\psi\big(-\nabla^2\delta\bar\psi + R^{-2}_\alpha\delta\bar\psi\big)\,d^2x\nonumber\\ &=  \int_D \bar\psi\left(R^{-2}_\alpha \delta\nu - \delta\bar\xi\right)\,d^2x\nonumber\\ &= \left\langle -\bar\psi,\delta\bar\xi\right\rangle + \left\langle R^{-2}_\alpha\bar\psi, \delta\psi_\sigma\right\rangle - \sum_{n=1}^\alpha \left\langle R^{-2}_\alpha (n+1) \overline{\sigma^{n+1}}\bar\psi, \delta\psi_{\sigma^n}\right\rangle,
  \end{align}
  as above but this time taking into account Assump.\@~\eqref{ass:cir-var} and the definition of $\nu$, given in \eqref{eq:nu}.  Thus we get
  \begin{equation}
	 \frac{\delta \mathscr H}{\delta\bar\xi} = -\bar\psi,\quad \frac{\delta \mathscr H}{\delta\psi_\sigma} = R^{-2}_\alpha\bar\psi,\quad \frac{\delta \mathscr H}{\delta\psi_{\sigma^n}} = -R^{-2}_\alpha(n+1)\overline{\sigma^{n+1}}\bar\psi,
	 \label{eq:der}
  \end{equation}
  $n = 1, 2, \dotsc, \alpha$, and consistently we compute
  \begin{equation}
	 \frac{\delta \mathscr L}{\delta\bar\psi} = \frac{\delta \langle -\bar\xi,\bar\psi \rangle}{\delta\bar\psi} - \frac{\delta \mathscr H}{\delta\bar\psi} = - \bar\xi
  \end{equation}
  since
  \begin{equation}
	 \mathscr H = \mathscr H(\bar\xi,\psi_\sigma,\dotsc,\psi_{\sigma^{\alpha+1}}).
  \end{equation}
  Finally, we compute
  \begin{align}
	 \dot{\mathscr F} {}=& \left\langle\frac{\delta \mathscr F}{\delta\bar\xi}, \partial_t\bar\xi\right\rangle + \sum_{n=1}^{\alpha+1} \left\langle\frac{\delta \mathscr F}{\delta\psi_{\sigma^n}}, \partial_t\psi_{\sigma^n}\right\rangle\nonumber\\ {}=& - \left\langle\frac{\delta \mathscr F}{\delta\bar\xi},[\bar\psi,\bar\xi - R^{-2}_\alpha\nu]\right\rangle - \sum_{n=1}^{\alpha+1} \left\langle\frac{\delta \mathscr F}{\delta\psi_{\sigma^n}}, [\bar\psi, \psi_{\sigma^n}]\right\rangle \nonumber\\ {}=& \left\langle\frac{\delta \mathscr F}{\delta\bar\xi}, \left[\frac{\delta \mathscr H}{\delta\bar\xi},\xi\right] + R^{-2}_\alpha\left[\left( 1 - \sum_{n=1}^\alpha (n+1) \overline{\sigma^{n+1}}\right)\bar\psi,\psi_{\sigma^n}\right]\right\rangle\nonumber\\ {}&+ \sum_{n=1}^{\alpha+1} \left\langle \frac{\delta \mathscr F}{\delta\psi_{\sigma^n}}, \left[\frac{\delta \mathscr H}{\delta\bar\xi},\psi_{\sigma^n}\right]\right\rangle\nonumber\\ {}=& \left\langle\frac{\delta \mathscr F}{\delta\bar\xi}, \left[\frac{\delta \mathscr H}{\delta\bar\xi},\xi\right] + \sum_{n=1}^{\alpha+1} \left[\frac{\delta \mathscr H}{\delta\psi_{\sigma^n}},\psi_{\sigma^n}\right]\right\rangle + \sum_{n=1}^{\alpha+1} \left\langle \frac{\delta \mathscr F}{\delta\psi_{\sigma^n}}, \left[\frac{\delta \mathscr H}{\delta\bar\xi},\psi_{\sigma^n}\right]\right\rangle,
  \end{align}
  where we have used, in order: 1) the IL$^{(0,\alpha)}$QG equations \eqref{eq:ILQG}; 2) the functional derivative of $\mathscr H$ with respect to $\bar\xi$, given in \eqref{eq:der}; 3) the definition of $\nu$, given in \eqref{eq:nu}; and 4) the functional derivative of $\mathscr H$ with respect to $\psi_{\sigma^n}$, $n = 1, 2, \dotsc, \alpha+1$, given in \eqref{eq:der}.  Equation \eqref{eq:LP} finally follows by the skew-adjointness of $[a,\cdot\hspace{.01cm}]$ for all $a \in \mathcal F(D)$ (with respect to the $L^2$ pairing), guaranteed by the admissibility condition \eqref{eq:adm-IL}.
\end{proof}

The Jacobian $[\,,\hspace{.01cm}]$ satisfies $[a,b] = -[b,a]$ (antisymmetry) and $[a,[b,c]] + \operatorname{\circlearrowleft} = 0$ (Jacobi identity) for every $a,b,c \in \mathcal F(D)$.  These two properties are inherited by the bracket $\{\,,\hspace{.01cm}\}$ in \eqref{eq:LP} by its linearity in $(\bar\xi, \psi_\sigma, \psi_{\sigma^2}, \dotsc, \psi_{\sigma^{\alpha+1}})$. An explicit proof is given in \cite{Beron-21-POFb}.  The latter makes $\{\,,\hspace{.01cm}\}$ a \emph{Lie--Poisson bracket}.  More concretely, letting $\mu := (\bar\xi, \psi_\sigma, \psi_{\sigma^2}, \dotsc, \psi_{\sigma^{\alpha+1}})$, we can write $\{\,,\hspace{.01cm}\}$ as
\begin{equation}
  \{\mathscr F, \mathscr H\}(\mu) = \left\langle\mu,\left[\frac{\delta \mathscr F}{\delta\mu}, \frac{\delta \mathscr H}{\delta\mu}\right]_\ltimes\right\rangle,
  \label{eq:LP-semidir}
\end{equation}
where
\begin{equation}
  [\mathbf a,\mathbf b]_\ltimes := ([a_1,b_1], [a_1,b_2]-[a_2,b_1], [a_1,b_3]-[a_3,b_1],\dotsc,[a_1,b_{\alpha+2}]-[a_{\alpha+2},b_1])
  \label{eq:semidir}
\end{equation}
with $\mathbf a := (a_1,a_2,\dotsc,a_{\alpha+2})$ and $\mathbf b := (b_1,b_2,\dotsc,b_{\alpha+2})$ for $a_i,b_i \in \mathcal F(D)$. The bracket $[\,,\hspace{.01cm}]_\ltimes$ is antisymmetric and satisfies the Jacobi identity, properties inherited from $[\,,\hspace{.01cm}]$, which in turn are transferred to $\{\,,\hspace{.01cm}\}$ by its linearity in $\mu$.  Operationally, to obtain the IL$^{(0,\alpha)}$QG it is the dual (with respect to the $L^2$ pairing) of the skew-adjoint operator $[\mu,\cdot\,]_\ltimes$ that is needed, that is, $ -[\mu,\cdot\,]_\ltimes =:\mathbb J(\mu)$. This known as the \emph{Poisson operator}, which can be written compactly as
\begin{equation}
  \mathbb J^{nm} = 
  \begin{pmatrix}
	 -[\bar\xi,\cdot\,] & -[\psi_{\sigma^m},\cdot\,]\\ -[\psi_{\sigma^n},\cdot\,] & 0
  \end{pmatrix},
  \label{eq:J}
\end{equation}
$n = 1,2,\dotsc,\alpha+1$. The IL$^{(0,\alpha)}$QG as a generalized Hamiltonian system of Lie--Poisson type then follows as
\begin{equation}
  \renewcommand{\arraystretch}{1.5}
  \begin{pmatrix} 
	 \partial_t\bar\xi\\ \partial_t\psi_{\sigma^{n}} 
  \end{pmatrix} = \sum_{m=1}^n\mathbb J^{nm} 
  \begin{pmatrix} 
	 \frac{\delta \mathscr H}{\delta\bar\xi}\\ \frac{\delta \mathscr H}{\delta\psi_{\sigma^m}} 
  \end{pmatrix}.
\end{equation}

The formula for the Hamiltonian of the IL$^{(0,\alpha)}$QG, given in \eqref{eq:H}, defines a \emph{partial Legendre transform} $(\bar\psi, \psi_\sigma, \psi_{\sigma^2}, \dotsc, \psi_{\sigma^{\alpha+1}}) \mapsto (\bar\xi, \psi_\sigma, \psi_{\sigma^2}, \dotsc, \psi_{\sigma^{\alpha+1}})$. This allows the Lie--Poisson Hamiltonian structure of the IL$^{(0,\alpha)}$QG, deduced in \cite{Beron-21-POFb} by direct manipulation, to be obtained from its Euler--Poincar\'e variational formulation, derived here. Informally, $\smash{\frac{\delta}{\delta\bar\psi}\mathscr L = -\bar\xi}$ may be seen as a momentum conjugate to $\bar\psi$, which justifies viewing \eqref{eq:H} as a (partial) Legendre transform. A rigorous interpretation of \eqref{eq:H} as a Legendre transform necessitates the incorporation of specific geometric mechanics notions. These are elaborated upon in Sec.\@~\ref{sec:geomech}, with a particular emphasis on Sec.\@~\ref{sec:mom}.

\begin{remark}\label{rem:ilqgm}
  The ILQGM \eqref{eq:ILQGM} in the variables
  $(\bar\zeta:=\kappa^{-1}\bar\xi,\psi_\sigma)$ follows from
  \begin{equation}
	 \dot{\mathscr F} = \{\mathscr F,\mathscr H\}_\mathrm{ILQGM} := \int_D\kappa\bar\zeta \left[\frac{\delta\mathscr F}{\delta\bar\zeta}, \frac{\delta\mathscr H}{\delta\bar\zeta}\right] + \kappa\psi_\sigma \left(\left[\frac{\delta\mathscr F}{\delta\bar\zeta},\frac{\delta\mathscr H}{\delta\psi_\sigma}\right] + \left[\frac{\delta\mathscr F}{\delta\psi_\sigma},\frac{\delta\mathscr H}{\delta\bar\zeta}\right]\right)\,d^2x
	 \label{eq:ILQGMbra}
  \end{equation}
  for any $\mathscr F(\bar\zeta,\psi_\sigma)$ with
  \begin{equation}
	 \mathscr H(\bar\zeta,\psi_\sigma) = \frac{1}{2}\int_D \kappa|\nabla\bar\psi|^2 + R^{-2}\bar\psi\,d^2x.
  \end{equation}
  The bracket $\{\,,\hspace{.01cm}\}_\mathrm{ILQGM}$ is antisymmetric. However, it does not satisfy the Jacobi identity unless $\kappa$ is taken to be a constant.  Specifically, using $\mathscr F_{\bar\zeta}$ as a shorthand for $\smash{\frac{\delta}{\delta\bar\zeta}}\mathscr F$, we have that $ \{\mathscr F, \mathscr H\}_{\bar\zeta} = \kappa [\mathscr F_{\bar\zeta}, \mathscr H_{\bar\zeta}]$ plus second-order terms, which can be shown \textup{\cite{Morrison-82}} to neglibly contribute to the Jacobi identity by the skew-adjointness (with respect to the $L^2$ pairing) of the ``Poisson'' operator $-[\kappa\bar\zeta,\cdot\,]$. To see that $\{\,,\hspace{.01cm}\}_\mathrm{ILQGM}$ fails to satisfy the Jacobi identity, the first term in \eqref{eq:ILQGMbra} is enough to be inspected.  Denote it by $\{\mathscr F, \mathscr H\}^{\bar\zeta}$.  Using the fact that the canonical Poisson bracket $[\,,\hspace{.01cm}]$ saisfies the Jacobi identity, one computes $\{\{\mathscr F, \mathscr G\}^{\bar\zeta},\mathscr H\}^{\bar\zeta} + \operatorname{\circlearrowleft} = \int_D \kappa\bar\zeta([\mathscr F_{\bar\zeta},\mathscr G_{\bar\zeta}][\kappa, \mathscr H]  + \operatorname{\circlearrowleft})\, d^2x$, which vanishes if and only if $\kappa$ is a constant.  Thus the ILQGM does not represent a Hamiltonian system.  It might be classified though as a ``quasi'' Hamiltonian system according to the definition of \textup{\cite{Dubos-Tort-14}}.
\end{remark}

\section{Conservation laws}\label{sec:conslaws}

The antisymmetry of the Lie--Poisson bracket \eqref{eq:LP} implies the conservation of energy: $\dot{\mathscr{H}} = \{\mathscr{H}, \mathscr{H}\} = 0$. The conservation of $\mathscr{H}$ can be linked with the invariance of $\mathscr{H}$ itself under time translations as a result of Noether’s theorem.  Specifically, let $\mathscr G(\mu)$ be the generator of an infinitesimal transformation defined by $\delta_{\mathscr G} := -\varepsilon\{\mathscr G,\cdot\}$ where $\varepsilon \downarrow 0$ \cite{Shepherd-90}.  The (infinitesimal) action of $\mathscr G$ on any functional $\mathscr F(\mu)$ is given by
\begin{equation}
    \Delta_{\mathscr G}\mathscr F \sim -\varepsilon \{\mathscr G,\mathscr F\}.
    \label{eq:D_GF}
\end{equation}
By setting $\mathscr F = \mathscr H$, it follows that a symmetry of the Hamiltonian implies a conservation law and vice versa, which is an expression of Noether's theorem. Clearly, $\mathscr H$ is the generator of time shifts $t \to t + \varepsilon$ since $\delta_{\mathscr H} \mu = \varepsilon \partial_t \mu$. Specialized to the IL$^{(0, \alpha)}$QG, conservation of energy is linked to symmetry of the IL$^{(0, \alpha)}$QG's Hamiltonian under time shifts. For this to be fully self-consistent, Assump.\@~\ref{ass:cir-cons} on the preservation of the velocity circulation along the flow domain boundary is key. Indeed, by direct manipulation of the IL$^{(0, \alpha)}$QG system \eqref{eq:ILQG}, that is, upon multiplying the equation for $\bar{\xi}$ by $\bar{\psi}$ and integrating over $D$, one finds that
\begin{equation}
  \dot{\mathscr{H}} = \bar{\psi}\vert_{\partial D}\dot{\gamma},
\end{equation} 
which vanishes provided that $\dot{\gamma} = 0$. 

Referring back to \eqref{eq:D_GF}, the calculation quickly proceeds to
\begin{equation}
    \frac{d}{dt}\Delta_{\mathscr G}\mathscr F - \Delta_{\mathscr G}\frac{d}{dt}\mathscr F \sim\varepsilon \left\{\mathscr F,\frac{d}{dt}\mathscr G\right\}.
    \label{eq:dGFdt}
\end{equation}
From this, it follows that if the generator of the transformation is conserved, it produces a symmetry in the most general sense: allowing time to pass and executing a transformation are operations that commute \cite{Ripa-RMF-92a}.

The reciprocal of the above quite general Noether's theorem is not true: the noted general symmetry implies that the generator of the symmetry is an arbitrary function of \emph{Casimirs} $\mathscr{C}(\mu)$, satisfying \cite[cf., e.g.,][]{Marsden-Ratiu-99}
\begin{equation}
    \{\mathscr C,\mathscr F\} = 0\quad \forall \mathscr F(\mu).
\end{equation}
Because the Casimirs Poisson-commute with any functional, they are conserved. An important observation is that the Casimirs do \emph{not} produce any transformation. Their conservation is still connected to symmetries, but these are not visible in the Eulerian description of fluid flow. We will return to this in the following section.

The Casimirs for $\alpha > 0$, derived in \cite{Beron-21-POFb}, are given by
\begin{equation}
  \mathscr C^\alpha_{a,F} := \int_D a\bar\xi + F(\psi_\sigma, \psi_{\sigma^2}, \dotsc, \psi_{\sigma^{\alpha+1}})\,d^2x
  \label{eq:CF}
\end{equation}
for any constant $a$ and function $F$. The Casimir for $\alpha = 0$, i.e., the IL$^0$QG, is
\begin{equation}
  \mathscr C^0_{F,G} := \int_D \bar\xi F(\psi_\sigma) + G(\psi_\sigma)\,d^2x,
  \label{eq:CFG}
\end{equation}
where $F,G$ are arbitrary function. These Casimirs have been known to exist for some time since the IL$^0$QG, incompressible Euler–Boussinesq flow on a vertical plane \cite{Benjamin-86}, and the so-called low-$\beta$ reduced magnetohydrodynamics \cite{Morrison-Hazeltine-84}, all share the same bracket.  The ILQGM discussed in \cite{Ripa-RMF-96} also supports this conservation law which commutes with any function in a bracket which however does no satisfy the Jacobi identity; cf.\ Rem.\@~\ref{rem:ilqgm}.  Finally, for completeness we write down Casimir for the HLQG:
\begin{equation}
  \mathscr C_F : = \int_D F(\bar\xi)\,d^2x,
  \label{eq:CHL}
\end{equation}
where $F$ is any function. The Casimir has a well-documented historical linage \cite{Morrison-81, Weinstein-83}.

\begin{remark}
  Note that, in a broad sense, $\mathscr C^\alpha_{1,0} = \mathscr C^0_{1,0} = \mathcal K(D)$, which represents the Kelvin--Noether circulation along the boundary of the flow domain \eqref{eq:KD}. Likewise, $\mathscr C^\alpha_{0,\nu} = C^0_{0,\nu} = \int_D \bar{\psi} \, d^2x - \mathcal V$, which is related to volume conservation; recall Rem.\@~\ref{rem:vol}. The apparent looseness arises from the fact that the IL$^{(0,\alpha)}$QG model considers fewer advected buoyancies as $\alpha$ approaches 0, thus these equalities should not be interpreted as strict equalities.
\end{remark}

\section{Particle relabeling symmetry and Casimir conservarion}\label{sec:Noether}

\begin{lemma}
  The IL$^{(0,\alpha}$QG preserves, in addition to the Casirmirs, the following quantities:
  \begin{equation}
    \mathscr I_F(\bar\xi,\nu) := \int_D\bar\xi F(\nu)\,d^2x,
    \label{eq:I}
  \end{equation}
where $\nu$ is the linear combination of the buoyancy coefficients defined in \eqref{eq:nu}.   
\end{lemma}

\begin{proof}
  The proof is a trivial extension of that given by \cite{Beron-Olascoaga-24} for the particular case of the IL$^{(0,1)}$QG. 
\end{proof}

Building upon the analysis presented in \cite{Beron-Olascoaga-24}, with the exception of the vertically mixed scenario ($\alpha = 0$), the conservation laws \eqref{eq:I} are not Casimirs of the Lie--Poisson bracket \eqref{eq:LP}.  However, they form the kernel of the following bracket:
\begin{equation}
	 \{\mathscr F,\mathscr G\}_\nu := \left\langle\bar\xi,\left[\frac{\delta \mathscr F}{\delta\bar\xi},\frac{\delta \mathscr G}{\delta\bar\xi}\right]\right\rangle + \left\langle\nu, \left[\frac{\delta \mathscr F}{\delta\bar\xi},\frac{\delta \mathscr G}{\delta\nu}\right] + \left[\frac{\delta \mathscr F}{\delta\nu},\frac{\delta \mathscr G}{\delta\bar\xi}\right] \right\rangle.
	 \label{eq:LP-nu}
\end{equation}
(To be more precise, the Casimirs of the above bracket are given by \eqref{eq:I} plus $\int_D G(\nu)\,d^2x$ where $G$ is an arbitrary function.) Upon evaluating this bracket with the Hamiltonian given by equation \eqref{eq:H}, or more explicitly in \eqref{eq:H-exp}, understood as a functional of $(\bar\xi, \nu)$, the following set of equations of motion emerges:
\begin{equation}
  \partial_t \bar\xi + [\bar\psi,\bar\xi] =  R^{-2}_\alpha[\bar\psi,\nu],\quad \partial_t \nu + [\bar\psi,\nu] = 0.
  \label{eq:ILQG-nu}
\end{equation}
This system governs the evolution of the potential vorticity $\bar\xi$, which remains unaffected by the details of the dynamics of the individual buoyancy coefficients $\psi_{\sigma}^n$, where $n = 1, 2, \dotsc, \alpha+1$. Each of these coefficients is independently conserved in a material manner.  Observe that the ``bulk'' dynamics described by \eqref{eq:ILQG-nu} are identical (modulo a difference in the Rossby deformation scale) to those of the IL$^0$QG.

The conservation laws \eqref{eq:I} represent Casimirs but in a weaker sense than \eqref{eq:CF} and \eqref{eq:CFG} as they pertain to the bulk IL$^{(0,\alpha}$QG potential vorticity dynamics. In any case, these integrals of motion are all related, via Noether theorem, to particle relabeling symmetry, as we proceed to demonstrate next.

First, we note that Euler--Poincar\'e system dual to \eqref{eq:ILQG-nu} follows from the Euler--Poincar\'e variational principle \eqref{eq:dS} with the Lagrangian understood as a functional of $(\bar\psi, \nu)$ and with constraints \eqref{eq:cons} replaced by
\begin{equation}
    \delta\bar\psi = \partial\eta + [\bar\psi,\eta],\quad \delta\nu = -[\eta,\nu].
    \label{eq:cons-nu}
\end{equation}

Next, observe that the streamfunction, $\bar{\psi}$, and the buoyancy coefficients, $\psi_{\sigma}^{n}$, $n = 1, 2, \dotsc, \alpha + 1$, \emph{remain invariant under a relabeling of fluid particle labels}. This invariance can be rigorously articulated using the language of differential geometry as adopted in Sec.\@~6. The relabeling of the particles leaves the IL$^{(0,\alpha)}$QG Lagrangian \eqref{eq:L} unchanged:
\begin{equation}
    \delta\mathscr L = \left\langle\frac{\delta\mathscr L}{\delta\bar\psi},\delta\bar\psi\right\rangle + \sum_{n=1}^{\alpha+1}\left\langle\frac{\delta\mathscr L}{\delta\psi_{\sigma}^n},\delta\psi_{\sigma}^n\right\rangle = 0
    \label{eq:dL}
\end{equation}
since
\begin{equation}
    \delta\bar\psi = 0,\quad \delta\psi_{\sigma}^n = 0,
    \label{eq:d0}
\end{equation}
$n = 1,2,\dotsc,\alpha+1$, under the relabeling.  Comparing \eqref{eq:d0} with \eqref{eq:cons} it follows that
\begin{equation}
    \partial_t \eta + [\bar\psi,\eta] = 0,\quad
	 [\eta,\delta\psi_{\sigma}^n] = 0, \label{eq:eta-evol}
\end{equation}
$n = 1,2,\dotsc,\alpha+1$, where $\eta$ is interpreted as the \emph{generator of the (Lie) symmetry of the Lagrangian} \eqref{eq:dL}. When $\alpha = 0$, the generator is given by $\eta = F(\psi_\sigma)$, where $F$ is an arbitrary function. This is because $[F(\psi_\sigma), \psi_\sigma] = F' \nabla^\perp \psi_\sigma \cdot \nabla \psi_\sigma = 0$ and $\partial_t F(\psi_\sigma) + [\bar\psi, F(\psi_\sigma)] \equiv \frac{D}{Dt} F(\psi_\sigma) = F' \frac{D}{Dt} \psi_\sigma = 0$. In the stratified ($\alpha > 0$) case, however, the generator must be a constant since only $\eta = \mathrm{const}$ can simultaneously satisfy \eqref{eq:eta-evol}.

The particle relabeling map equally preserves the Lagrangian for the Euler--Poincaré dynamics involving the variables $(\bar\psi, \nu)$. In a manner akin to the $\alpha = 0$ case, the generator $\eta$ of the corresponding symmetry is an arbitrary function of $\nu$, since such an $\eta$ satisfies
\begin{equation}
    \partial_\eta + [\bar\psi,\eta] = 0,\quad [\eta,\nu] = 0,
    \label{eq:eta-nu}
\end{equation}
exactly.  It is evident that $\eta = F(\nu)$, where $F$ is an arbitrary function, can only satisfy \eqref{eq:eta-evol} in a weak sense.

Now, computing the variation of the Euler--Poincar\'e action for the L$^{(0,\alpha)}$QG as in the proof of Thm.\@~\ref{thm:EPforILQG} but lifting up the restriction that the variations $\delta\bar\psi$ and $\delta\psi_{\sigma}^n$, $n = 1,2,\dotsc,\alpha+1$, vanish at the endpoints of the integration, a boundary term emerges, given by $J^\eta\vert_{t_0}^{t_1}$ where
\begin{equation} 
    J^\eta := \left\langle\frac{\delta\mathscr L}{\delta\bar\psi},\eta\right\rangle.
\end{equation}
It follows that to fulfill the Euler--Poincar\'e variational principle, \emph{the quantity $J^\eta$ must remain constant along the dynamics produced by the Euler--Poincar\'e system for the IL$^{(0,\alpha)}$QG}, given by \eqref{eq:EPforILQG}. Conservation of $J^\eta$ is related to symmetry since 
\begin{align}
    \dot J^\eta &= \left\langle\partial_t\frac{\delta\mathscr L}{\delta\bar\psi},\eta\right\rangle + \left\langle\frac{\delta\mathscr L}{\delta\bar\psi},\partial_t\eta\right\rangle\nonumber\\ &= \left\langle - \left[\bar\psi,\frac{\delta \mathscr L}{\delta\bar\psi}\right] - \sum_{n=1}^{\alpha+1} \left[\frac{\delta\mathscr L}{\delta\psi_{\sigma^n}},\psi_{\sigma^n}\right],\eta\right\rangle + \left\langle\frac{\delta\mathscr L}{\delta\bar\psi},\delta\bar\psi - [\bar\psi,\eta]\right\rangle\nonumber\\ &= \left\langle\frac{\delta\mathscr L}{\delta\bar\psi},\delta\bar\psi\right\rangle + \sum_{n=1}^{\alpha+1}\left\langle\frac{\delta\mathscr L}{\delta\psi_{\sigma}^n},\delta\psi_{\sigma}^n\right\rangle\nonumber\\ &= \delta\mathscr L,
\end{align}
where Euler--Poincar\'e equation \eqref{eq:EP} and the constraints \eqref{eq:cons} were employed along with \eqref{eq:dual}.  The above expression vanishes under particle relabeling \eqref{eq:d0}. Consequently, the quantity $J^\eta$ can be appropriately referred to as a \emph{Noether quantity}. For $\alpha = 0$, with the symmetry generator given by $\eta = -F(\psi_\sigma)$, where $F$ is arbitrary just as is the sign, and noting that $\frac{\delta}{\delta\psi}\mathscr L = -\bar\psi$, cf.\@~\eqref{eq:dLdpsi}, it follows that
\begin{equation}
    J^\eta = \int_D\bar\xi F(\psi_\sigma)\,d^2x.
\end{equation}
The critical observation is that the above $J^\eta$ gives the first term of the IL$^0$QG Casimir \eqref{eq:CF}. Noting that the second term follows immediately by material conservation of $\psi_\sigma$, a connection between this Casimir and particle relabeling symmetry via the Noether theorem is established. Similar connections follow for the remaining Casimirs. When $\alpha > 0$ the symmetry generator can be taken to be, with no loss of generality, $\eta = -a$, where $a$ is an arbitrary constant. Then,
\begin{equation}
    J^\eta = a\int_D\bar\xi\,d^2x,
\end{equation}
which gives the first term of the IL$^{(0,\alpha)}$QG, $\alpha>0$, Casimir \eqref{eq:CF}.  

Finally, the reasoning in the preceding paragraph can be applied to the variables $(\bar{\psi}, \nu)$. Recalling that the generator of the particle relabeling symmetry in such variables is $\eta = - F(\nu)$, where sign and function $F$ are arbitrary, it follows that the Noether quantity is given by
\begin{equation}
    J^\eta = \int_D\bar\xi F(\nu)\,d^2x,
\end{equation}
which coincides with the conservation law \eqref{eq:I}, thereby relating it to particle relabeling symmetry through Noether's theorem.

In conclusion and to ensure completeness, we examine the HLQG given that the explicit associations between Casimir--Noether quantities have not yet been established, to the best of our knowledge. The generator $\eta$ of particle relabeling symmetry satisfies in this case satisfies the following evolution equation:
\begin{equation}
    \partial_t\eta + [\bar\psi,\eta] = 0.
\end{equation}
This equation is readily satisfied by any function of $\bar\xi$. By selecting $\eta = -\bar\xi^{-1}F(\bar\xi)$ for an arbitrary function $F$, we deduce that
\begin{equation}
    J^\eta = \int_D F(\bar\xi)\,d^2x,
\end{equation}
which precisely corresponds to the HLQG Casimir \eqref{eq:CHL}.

\section{Geometric mechanics interpretation}\label{sec:geomech}

The language of diffeomorphisms in differential geometry is the appropriate one to rigorously communicate the results discussed so far. In this section, we adopt this language, which further enables us to shed additional light on several aspects.

\subsection{Geometric view of IL$^{(0,\alpha)}$QG
flow}\label{sec:geoflow}

First, recall that the fluid domain $D\subset\mathbb{R}^2$. Then, geometrically \cite[e.g.,][]{Marsden-Ratiu-99,Holm-etal-09}, the Euler--Poincar\'e variational principle is defined on $\mathfrak{sdiff}(D) \times \mathcal F(D)^{\alpha+1}$, namely, the Cartesian product of the \emph{Lie algebra of the group of symplectic (i.e., area-preserving) diffeomorphisms} on $D$, denoted $\mathrm{SDiff}(D)$, with $\alpha+1$ Cartesian copies of the space of smooth time-dependent scalar fields on $D$.  This involves to first view trajectories of fluid particles on $D$ as curves on $\mathrm{SDiff}(D)$, which is achieved by lifting of the motion on $D$ to $\mathrm{SDiff}(D)$, whose \emph{left-action} on $D$ produces the fluid trajectories (Fig.\@~\ref{fig:geo}).

\begin{figure}[t!]
  \centering%
  \includegraphics[width=.5\textwidth]{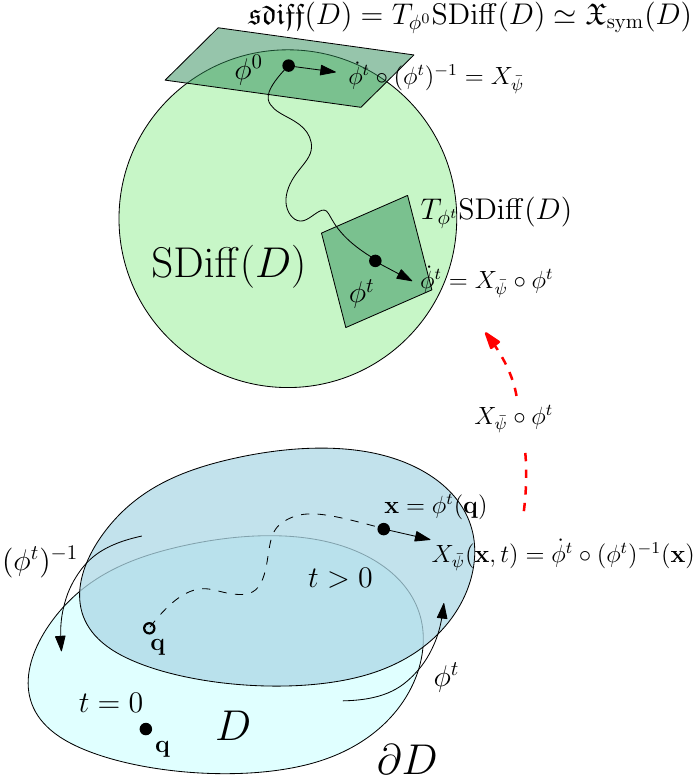}%
  \caption{Geometric mechanics view of fluid motion produced by the IL$^{(0,\alpha)}$QG. See text for details.}
  \label{fig:geo}%
\end{figure}

Specifically, the set in $\mathrm{SDiff}(D)$ is given by $\{\phi^t\}$, where $\phi^t : D\to D$; $\mathbf q \mapsto \mathbf x$ is a smoothly invertible area preserving map of positions of fluid particles at time $t=0$, representing the \emph{reference configuration}, to positions at time $t > 0$, representing the \emph{current configuration}.  This expresses the left-action of $\mathrm{SDiff}(D)$ on $D$.  Explicitly, $\mathbf x = \phi^t(\mathbf q)$, which gives the trajectory of a fluid trajectory starting at $\mathbf q$, taken as a \emph{label}. The diffeomorphic nature of $\phi^t$ conveys to $\mathrm{SDiff}(D)$ a differentiable (manifold) structure and hence a Lie character.  Multiplication in $\mathrm{SDiff}(D)$ is given by the composition of functions. Thus, $\mathrm{SDiff}(D)$ is given by the pair $(\{\phi^t\},\circ)$. The vector space corresponding to $\mathfrak{sdiff}(D)$ is identified with the tangent space to $\mathrm{SDiff}(D)$ at the identity, arranged to happen at $t=0$, namely, $\phi^0$. This point is special inasmuch as any point on $\mathrm{SDiff}(D)$ can be accessed via $\phi^t$. The bracket in $\mathfrak{sdiff}(D)$ is given by $[\,,\hspace{.01cm}]$, the \emph{canonical Poisson bracket} in $\mathbb R^2$, as follows.  The fluid velocity
\begin{equation}
  \nabla^\perp\bar\psi\cdot\nabla = \partial_x\bar\psi\partial_y -\partial_y\bar\psi\partial_x
\end{equation}
is obtained by \emph{right-translation} of $\dot\phi^t \in T_{\phi^t}\mathrm{SDiff}(D)$ with the inverse of $\phi^t  \in \operatorname{{SDiff}}(D)$, $(\phi^t)^{-1} \in \mathrm{{SDiff}}(D)$, to $T_{\phi^0}\mathrm{SDiff}(D)$:
\begin{equation}
  \nabla^\perp\bar\psi(\mathbf x,t)\cdot\nabla = \partial_t\vert_{\mathbf q}\phi^t(\mathbf q) =: \dot\phi^t(\mathbf q) = \dot\phi^t\circ (\phi^t)^{-1}(\mathbf x) \in T_{\phi^0}\mathrm{SDiff}(D).
\end{equation}
Therefore, $\mathfrak{sdiff}(D)$ can be conceptualized as the coset $T\mathrm{SDiff}(D)\backslash\mathrm{SDiff}(D)$, or equivalently $TD\backslash\mathrm{SDiff}(D)$. In other words, it can be seen as the collection of equivalence classes, where two elements of $T\mathrm{SDiff}(D)$ are considered equivalent if they differ by right multiplication by an element of $\mathrm{SDiff}(D)$.

\begin{remark}\label{rem:sob}
    A cautionary note is that, for the above geometric interpretation of fluid motion to be rigorously valid in the sense of \textup{\cite{Ebin-Marsden-70}}, $\mathrm{SDiff}(D)$ must be chosen from a specific Sobolev class, $s$, meaning that $\phi^t$ and its weak derivatives up to order $s$ should belong to $L^2(D)$. This condition would constitute $\mathrm{SDiff}(D)$ as a Hilbert manifold, equipped with inverse and implicit function theorems, as well as a general solution theorem for differential equations. When $s > 2$, only the right action of $\mathrm{SDiff}(D)$ on $T\mathrm{SDiff}(D)$ would be smooth; nonetheless, this is the operation of relevance for fluid dynamics.
\end{remark}
  
Now, since $\nabla\cdot\nabla^\perp\bar\psi = 0$, it follows that $\nabla^\perp\bar\psi\cdot\nabla$ represents a \emph{nonautonomous canonical Hamiltonian vector field}.  The corresponding Hamiltonian is the streamfunction, $\bar\psi$.  Denote, as usual,
\begin{equation}
  X_{\bar\psi} := \nabla^\perp\bar\psi\cdot\nabla
\end{equation}
and by $\mathfrak X_\mathrm{sym}(D)$ the space of (nonautonomous, canonical) Hamiltonian vector fields on $D$, viz., $X_{\bar\psi}\in \mathfrak X_\mathrm{sym}(D)$. The commutator of vectors is the operation that expresses the natural way in which elements of $T_{\phi^0}\mathrm{SDiff}(D)$, operationally identified with $\mathfrak X_\mathrm{sym}(D)$, act on themselves.  Such an operation is obtained by linearizing at the identity the \emph{left-conjugation}, namely,
\begin{equation}
  \left.\frac{d}{dt}\right\vert_{t=0} \left.\frac{d}{ds}\right\vert_{s=0} \phi^t_1\circ \phi^s_2\circ (\phi^t_1)^{-1} = X_{\bar\psi_1}X_{\bar\psi_2} - X_{\bar\psi_2}X_{\bar\psi_1} = \nabla^\perp[\bar\psi_1,\bar\psi_2] \cdot \nabla,
\end{equation}
where the last equality follows from cancellation of cross-derivatives. The last equality allows one to identify $[\,,\hspace{.01cm}]$ with the bracket of $\mathfrak{sdiff}(D)$.  One can then think of $\bar\psi$ as an element of $\mathfrak{sdiff}(D)$, understood as the pair $(\mathcal F(D), [\,,\hspace{.01cm}])$, by identifying $\bar\psi \in \mathcal F(D)$ with $X_{\bar\psi} \in T_{\phi^0}\mathrm{SDiff}(D) \simeq \mathfrak X_\mathrm{sym}(D)$.

This way, since
\begin{equation}
  \bar\psi\in \mathfrak{sdiff}(D),
\end{equation}
the Lagrangian
\begin{equation}
  \mathscr L \in C^\infty\big(\mathfrak{sdiff}(D)\times \mathcal F(D)^{\alpha+1}\big) : \mathfrak{sdiff}(D) \times \mathcal F(D)^{\alpha+1} \to \mathbb R. 
\end{equation}
Elements of $\mathfrak{sdiff}(D)^*$, the dual of $\mathfrak{sdiff}(D)$, are identified using the $L^2$ pairing $\langle\,,\hspace{.01cm}\rangle : \mathfrak{sdiff}(D)^* \times \mathfrak{sdiff}(D) \to \mathbb R$.  With this identification in mind,
\begin{equation}
  \smash{\frac{\delta\mathscr L}{\delta\bar\psi}} \in \mathfrak{sdiff}(D)^*.
\end{equation}

Geometric interpretations of the constraints \eqref{eq:cons} to which Hamilton's principle in Thm.\ \ref{thm:EPforILQG} is subject to are in order.  Fix $t$ and extend $\phi^t$, the flow generated by $X_{\bar\psi}$, to a curve $\epsilon \mapsto \phi^t(\epsilon)$. Let $X_\eta = \nabla^\perp\eta\cdot\nabla$ be defined by
\begin{equation}
  X_\eta(\mathbf x,t) := \left.\frac{d}{d\epsilon}\right|_{\epsilon=0}\phi^t(\epsilon)\circ(\phi^t)^{-1}(\mathbf x) =: \delta\phi^t\circ(\phi^t)^{-1}(\mathbf x) \in T_{\phi^0}\operatorname{SDiff}(D).
\end{equation}
We begin with the constraint on $\bar\psi$. By subtracting
\begin{equation}
  \frac{d}{dt}(\delta\phi^t)\circ(\phi^t)^{-1} = \frac{d}{dt}(X_\eta\circ\phi^t)\circ(\phi^t)^{-1} = \partial_t X_\eta + (\nabla^\perp\bar\psi\cdot\nabla) X_\eta
\end{equation}
from
\begin{equation}
  \delta\dot\phi^t\circ(\phi^t)^{-1} = \delta(X_{\bar\psi}\circ\phi^t)\circ(\phi^t)^{-1} = \delta X_{\bar\psi} + (\nabla^\perp\eta\cdot\nabla) X_{\bar\psi}
\end{equation}
with commutativity of differentiation in mind, one finds
\begin{equation}
  \delta X_{\bar\psi} = \partial_tX_\eta + X_{\bar\psi}X_\eta - X_\eta X_{\bar\psi},
\end{equation}
which reduces to the left equation in \eqref{eq:cons}.

Consider next the constraints on $\psi_{\sigma^n}$, $n = 1, 2, \dotsc, \alpha + 1$. Let $\psi_{\sigma^n}^t(\mathbf x) := \psi_{\sigma^n}(\mathbf x,t)$. Material conservation of $\psi_{\sigma^n}^t(\mathbf x)$ expresses as
\begin{equation}
  \psi_{\sigma^n}^t(\mathbf x) =  \psi_{\sigma^n}^0(\mathbf q) = (\phi_t)_*\psi_{\sigma^n}^0(\mathbf x) = ((\phi_t)^{-1})^*\psi_{\sigma^n}^0(\mathbf x),
\end{equation}
where $(\phi^t)_*$ denotes \emph{pushforward} by $\phi^t$.  Then one has
\begin{align}
  \delta \psi_{\sigma^n}^t := \left.\frac{d}{d\epsilon} \psi_{\sigma^n}^{t,\epsilon}\right\vert_{\epsilon=0} = \left.\frac{d}{d\epsilon} ((\phi_{t,\epsilon})^{-1})^*\psi_{\sigma^n}^0\right\vert_{\epsilon=0} &= \left.\left.\frac{d}{ds}\right\vert_{s=0} ((\phi_{t,\epsilon+s})^{-1})^*\psi_{\sigma^n}^0\right\vert_{\epsilon=0}\nonumber\\ &=  ((\phi_t)^{-1})^* \left.\frac{d}{d\epsilon}\right\vert_{\epsilon=0} ((\phi_{t,\epsilon})^{-1})^*\psi_{\sigma^n}^0\nonumber\\ &= - ((\phi_t)^{-1})^*\pounds_{X_\eta}\psi_{\sigma^n}^0\nonumber\\ &= - \pounds_{X_\eta}\psi_{\sigma^n}^t,
\end{align}
where the \emph{Lie derivative's dynamic definition} was used in the penultimate step. The right equations in \eqref{eq:cons} are obtained upon noting that, since $\psi_{\sigma^n}^t$ is a scalar, $\pounds_{X_\eta}\psi_{\sigma^n}^t$ simply is the derivative of $\psi_{\sigma^n}^t$ in the direction of $\nabla^\perp\bar\psi$.

A final concept is needed to complete the geometric mechanics interpretation of Thm.\@~\ref{thm:EPforILQG}.  A quantity is said to be \emph{advected} when it is dragged by Lie transport in the direction of the fluid velocity, say $u(\mathbf x,t)$, generalizing the notion of material conservation.  Let $a(\mathbf x,t) = a^t(\mathbf x)$ be an advected quantity.  This satisfies the following \emph{pullback} relationship:
\begin{equation}
	a^0(\mathbf q) = a^t(\mathbf x) = a^t\circ \phi^t(\mathbf q) = (\phi^t)^*a^t(\mathbf q).
\end{equation}
Taking the time derivative,
\begin{equation}
  0 = \frac{d}{dt}a^t(\mathbf x) = (\phi^t)^*(\partial_t + \pounds_v)a^t(\mathbf q) = (\partial_t + \pounds_v)a^t(\mathbf x),
\end{equation}
where the (second) \emph{Lie derivative theorem} was used. For $a\in \mathcal F(D)$ and $u = X_{\bar\psi}$, the above reads
\begin{equation}
  (\partial_t + \pounds_{X_{\bar\psi}})a = \partial_t a + \nabla^\perp\bar\psi\cdot\nabla a = \partial_t a + [\bar\psi,a] = 0,  
  \label{eq:adv}
\end{equation}
which is the equation satisfied by $\psi_{\sigma^n}$, $n = 1, 2,
\dotsc, \alpha + 1$.

With all the considerations above in mind, the set given by \eqref{eq:EPforILQG} and the equations for material conservation of $\psi_{\sigma^n}$, $n = 1,2,\dotsc,\alpha+1$, represent an \emph{Euler--Poincar\'e variational equation on symplectic diffeomorphisms with advected quantities}.

\subsection{Semidirect product Lie algebra}\label{rem:semidirect} 
 
The bracket $[\,,\hspace{.01cm}]_\ltimes$ in \eqref{eq:semidir} is a bracket for the algebra of the Lie group obtained by extending $\mathrm{SDiff}(D)$ by \emph{semidirect product} with $\mathcal F(D)^{\alpha+1}$, upon identifying the dual of $\mathcal F(D)$ with $\mathcal F(D)$ itself.  This follows by noting that the induced representation of $\mathfrak{sdiff}(D)$ on $\mathcal F(D)^{\alpha+1}$ is given by Lie differentiation with respect to canonical Hamiltonian vectors in $\mathbb R^2$, or, in term of functions, by canonical Poisson brackets in $\mathbb R^2$ \cite{Marsden-Morrison-84}.  The bracket of this semidirect product Lie algebra, denoted $\mathfrak{sdiff}(D) \ltimes \mathcal F(D)^{\alpha+1}$, carries the Lie--Poisson bracket $\{\,,\hspace{.01cm}\}$, given in \eqref{eq:LP} or more explicilty as written in \eqref{eq:LP-semidir}, on its dual, $\mathfrak{sdiff}(D)^* \ltimes \mathcal F(D)^{\alpha+1}$. The bracket $[\,,\hspace{.01cm}]_\ltimes$ is antisymmetric and satisfies the Jacobi identity, properties that are inherited by $\{\,,\hspace{.01cm}\}$, which in addition satisfies the Leibniz rule. Being $\{\,,\hspace{.01cm}\}$ a derivation in each of its arguments, it conveys to $C^\infty(\mathfrak{sdiff}(D)^* \ltimes \mathcal F(D)^{\alpha+1})$ a \emph{Lie enveloping algebra} structure, where $\mathfrak{sdiff}(D)^* \ltimes \mathcal F(D)^{\alpha+1}$ represents the underlying \emph{Poisson manifold}.

\subsection{Momentum map and Legendre transformation}\label{sec:mom}

The step that remains is to find the map that takes the cotangent bundle $T^*D$ to the dual of $\mathfrak{sdiff}(D)$ and the resulting Legendre transformation.  To accomplish this goal, we let $a := (\psi_\sigma, \psi_{\sigma^2}, \dotsc, \psi_{\sigma^{\alpha+1}})$ and construct a \emph{Clebsch action} for the dynamics produced by the Lagrangian $\mathscr L(\bar\psi,a)$ by constraining them to enforce the action of $\mathrm{SDiff}(D)$ on the fluid particle (i.e., Lagrangian) labels $\mathbf q \in TD$ and advect $a \in \mathcal F(D)^{\alpha+1}$, to wit,
\begin{equation}
  \mathfrak S = \int_{t_0}^{t_1} \mathscr L(\bar\psi,a)    + \big\langle \mathbf p,\partial_t \mathbf q + [\bar\psi,\mathbf q]\big\rangle_{TD} + \big\langle b, \partial_t a + [\bar\psi,a]\big\rangle_{\mathcal F(D)^{\alpha+1}}\,dt,
  \label{eq:action}
\end{equation}
where $\mathbf p$ and $b$ are Lagrange multipliers and we have labeled the angle brackets to make explicit the spaces they pair. Computing $\delta \mathfrak S = 0$ it follows that
\begin{equation}
  J^\eta := \big\langle\mathbf p,\delta \mathbf q\big\rangle_{TD} + \big\langle b,\delta a\big\rangle_{\mathcal F(D)^{\alpha+1}}
  \label{eq:J1}
\end{equation}
is constant along the dynamics produced by
\begin{equation}
  \begin{gathered}
	 \frac{\delta\mathscr L}{\delta\bar\psi} = [\mathbf p,\mathbf q] + [b,a],\quad \frac{\delta\mathscr L}{\delta a} = \partial_t b + [\bar\psi, b],\\ \partial_t\mathbf p + [\bar\psi,\mathbf p] = 0,\quad \partial_t\mathbf q + [\bar\psi,\mathbf q] = 0,\quad \partial_ta + [\bar\psi,a] = 0.
  \end{gathered}
  \label{eq:clebsch}
\end{equation}
A lengthy calculation shows that these dynamics coincide with those produced by the Euler--Poincar\'e equations \eqref{eq:EPforILQG}. The calculation involves evaluating $\langle\partial_t\smash{\frac{\delta}{\delta\bar\psi}}\mathscr L, \varphi\rangle_{\mathfrak{sdiff}(D)}$ for any $\varphi \in \mathfrak{sdiff}(D)$ following steps similar to those taken in Sec.\@~3.2 of \cite{Cotter-Holm-07}.

Consider now the \emph{particle relabeling map} $\mathbf q \mapsto r(\mathbf q)$ where $r$ is taken to a fixed element of $\mathrm{SDiff}(D)$.  The \emph{right-action} of the Lie group $R := (\{r\},\circ)$ on $\mathfrak{sdiff}(D)\times \mathcal F(D)^{\alpha+1}$ leaves it unchanged, i.e., $\delta\bar\psi = 0$ and $\delta a = 0$ under this action, which represents a continuous Lie symmetry. Indeed, with Sec.\@~\ref{sec:geoflow} in mind, we compute
\begin{equation}
  X_{\bar\psi}\cdot R = \frac{d}{dt}(\phi^{t}\circ r) \circ(\phi^t\circ r)^{-1} = \dot\phi^{t}\circ(r\circ r^{-1})\circ (\phi^t)^{-1} = X_{\bar\psi},
\end{equation}
so $\bar\psi$ remains unchanged under this action, and
\begin{equation}
  a^t\cdot R = (a^0\circ r)(g^t\circ r)^{-1} = a_0\circ(r\circ r^{-1})\circ (g^t)^{-1} =  a^t.
\end{equation}
In these circumstances, \eqref{eq:J1} reduces to
\begin{equation}
  J^\eta =  \big\langle\mathbf p,\delta \mathbf q\big\rangle_{TD}. 
  \label{eq:J2}
\end{equation}
The above \emph{relabeling symmetry} leaves $\mathscr L$
invariant.  Furthermore, since both $\mathbf q$ and $a$
represent advected quantities, this symmetry leaves the
Clebsch-constrained Lagrangian in \eqref{eq:action} equally unmodified. This provides a connection of conservation of $J^\eta$ in \eqref{eq:J2} with (relabeling) symmetry via Noether's theorem.  Note that
\begin{equation}
  J^\eta = \big\langle[\mathbf p,\mathbf q],\eta\big\rangle_{\mathfrak{sdiff}(D)} = \left\langle\frac{\delta\mathscr L}{\delta\bar\psi}, \eta\right\rangle_{\mathfrak{sdiff}(D)} + \big\langle [a,b], \eta\big\rangle_{\mathfrak{sdiff}(D)} =: J^\eta_1 + J^\eta_2, 
  \label{eq:J3}
\end{equation}
where we have first used that
\begin{equation}
  \delta\mathbf q = - [\eta,\mathbf q]	 
  \label{eq:dq}
\end{equation}
since $\mathbf q$ is an advected quantity, where $\eta = \delta\phi^t\circ (\phi^t)^{-1}$ (cf.\ Sec.\@~\ref{sec:geoflow}), and then the first equation in the top row of \eqref{eq:clebsch}. That $J^\eta$ indeed is preserved under the dynamics, i.e., it represents a \emph{Noether quantity}, can be verified directly.

\begin{proof}
  Note, on one hand, that,
  \begin{align}
	 \dot J^\eta_1 &= \left\langle\partial_t\frac{\delta\mathscr L}{\delta\bar\psi}, \eta\right\rangle_{\mathfrak{sdiff}(D)} + \left\langle\frac{\delta\mathscr L}{\delta\bar\psi}, \delta\bar\psi - [\bar\psi,\eta]\right\rangle_{\mathfrak{sdiff}(D)}\nonumber\\ &= \left\langle-\left[a,\frac{\delta\mathscr L}{\delta a}\right], \eta\right\rangle_{\mathfrak{sdiff}(D)} + \left\langle\frac{\delta\mathscr L}{\delta \bar\psi}, \delta\bar\psi\right\rangle_{\mathfrak{sdiff}(D)}\nonumber\\ &=  \left\langle\frac{\delta\mathscr L}{\delta a}, \delta a\right\rangle_{\mathcal F(D)^{\alpha+1}} + \left\langle\frac{\delta\mathscr L}{\delta \bar\psi}, \delta\bar\psi\right\rangle_{\mathfrak{sdiff}(D)}\\ &= \delta\mathscr L
  \end{align}
  and, on the other, that
  \begin{align}
	 \dot J^\eta_2  &= \left\langle\partial_t[a,b], \eta\right\rangle_{\mathfrak{sdiff}(D)} + \left\langle [a,b], \delta\bar\psi - [\bar\psi,\eta]\right\rangle_{\mathfrak{sdiff}(D)}\nonumber\\ &= \left\langle\partial_t[a,b] + \big[\bar\psi,[a,b]\big], \eta\right\rangle_{\mathfrak{sdiff}(D)} + \left\langle [a,b], \delta\bar\psi\right\rangle_{\mathfrak{sdiff}(D)}\nonumber\\ &= \left\langle \left[a, \frac{\delta\mathscr L}{\delta a}\right], \eta\right\rangle_{\mathfrak{sdiff}(D)} + \left\langle [a,b], \delta\bar\psi\right\rangle_{\mathfrak{sdiff}(D)}\nonumber\\ &= - \left\langle\frac{\delta\mathscr L}{\delta a}, \delta a\right\rangle_{\mathcal F(D)^{\alpha+1}} + \left\langle [a,b], \delta\bar\psi\right\rangle_{\mathfrak{sdiff}(D)}\\ &= - \delta\mathscr L + \left\langle [\mathbf p,\mathbf q], \delta\bar\psi\right\rangle_{\mathfrak{sdiff}(D)},
  \end{align}
  which follow by making use of: the constraints \eqref{eq:cons} on $\bar\psi$ and $\psi_{\sigma^n}$ ($n = 1, 2,\dotsc,\alpha+1$), or, equivalently, on $a$, given in \eqref{eq:dq}; the Euler--Poincar\'e equation for the IL$^{(0,\alpha)}$QG \eqref{eq:EPforILQG}, the equivalent relationships \eqref{eq:clebsch}; and the Jacobi identity satisfied by the canonical Poisson bracket $[\,,\hspace{.01cm}]$. Consequently,
  \begin{equation}
	 \dot J^\eta = \left\langle [\mathbf p,\mathbf q], \delta\bar\psi\right\rangle_{\mathfrak{sdiff}(D)} = 0
  \end{equation}
  by relabeling symmetry, as we wanted to verify.
\end{proof}

But since $\dot{J}^\eta_1 = 0$ and $\dot{J}^\eta_2 = 0$ each independently hold by relabeling symmetry, one can conveniently choose to write:
\begin{equation}
  J^\eta = \big\langle\mathbf p,\delta \mathbf q\big\rangle_{TD} =: \big\langle\mathbf J(\mathbf q,\mathbf p), \eta\big\rangle_{\mathfrak{sdiff}(D)}
\end{equation}
where
\begin{equation}
  \mathbf J(\mathbf q,\mathbf p) = [\mathbf p,\mathbf q] = \frac{\delta\mathscr L}{\delta \bar\psi} = -\bar\xi.
\end{equation}
Then, interpreting $\mathbf p(\mathbf x,t)$ as the \emph{canonical momentum} conjugate to $\mathbf q(\mathbf x,t)$, that is, the inverse map which tells what Lagrangian label occupies Eulerian position $\mathbf x$ at time $t$, the map
\begin{equation}
	\mathbf J : T^*D \to \mathfrak{sdiff}(D)^*  
\end{equation}
defines a \emph{momentum map} for the lift from the cotanget bundle of the manifold where the fluid is contained, $D$, to the dual of the Lie algebra of symplectic diffeomorphisms on $D$, where the IL$^{(0,\alpha)}$QG potential vorticity resides. 

The above provides a framework for constructing a \emph{partial Legendre transform} $(\bar\psi, a) \mapsto (\bar\xi, a)$ by pairing $\mathbf J(\mathbf q,\mathbf p) \in \mathfrak{sdiff}(D)^*$ with $\bar\psi \in \mathfrak{sdiff}(D)$.  This procedure is in fact used in \eqref{eq:H} to define the Hamiltonian of the IL$^{(0,\alpha)}$QG, which transforms the IL$^{(0,\alpha)}$QG from a set of Euler--Poincaré equations with advected quantities on $\mathfrak{sdiff}(D) \times \mathcal{F}(D)^{\alpha+1}$ into a Lie--Poisson system on $\mathfrak{sdiff}(D)^* \ltimes \mathcal{F}(D)^{\alpha+1}$.

We conclude by highlighting that the integral of the momentum map $\mathbf J(\mathbf q,\mathbf p)$ over $D$ corresponds to the Kelvin--Noether circulation (as discussed in Sec.\@~\ref{sec:kelvin}), elucidating the rationale behind its designation.

\subsection{Generalized IL$^{(0,\alpha)}$QG dynamics discovery}

We close by discussing a generalization of the Euler–Poincar\'e/Lie–Poisson framework that might already lead to more accurate dynamics than IL$^{(0,\alpha)}$QG dynamics.  Let $a\in \mathcal F(D)^n$, $n\in \mathbb N$.  Consider the Lagrangian $\mathscr L(\bar\psi,a) \in C^\infty(\mathfrak{sdiff}(D)\times \mathcal F(D)^n)$.  The constrained Hamilton's least action principle
\begin{equation}
  \delta\int_{t_0}^{t_1}\mathscr L(\bar\psi,a)\,dt = 0 : \delta\bar\psi = \partial_t\eta + [\bar\psi,\eta],\,\delta a = -[\eta,a],
\end{equation}
where $\eta\in \mathfrak{sdiff}(D)$ vanishes at the endpoints of integration, leads to the general Euler--Poincar\'e equations with advected quantities on $\mathfrak{sdiff}(D)\times \mathcal F(D)^n$:
\begin{equation}
  \partial_t\frac{\delta \mathscr L}{\delta\bar\psi} + \left[\bar\psi,\frac{\mathscr L}{\delta\bar\psi}\right] = \left[\frac{\delta\mathscr L}{\delta a},a\right],\quad \partial_t a + [\bar\psi,a] = 0.
  \label{eq:EP}
\end{equation}
Let $\mathscr L$ be such that $\frac{\delta}{\delta\bar\psi}\mathscr L = -\bar\xi$.  The corresponding Kelvin--Noether circulation theorem is
\begin{equation}
  \frac{d}{dt}\int_{D_t} \mathbf J\,d^2x = \int_{D_t} \left[\frac{\delta\mathscr L}{\delta a},a\right]\, d^2x,
\end{equation}
where 
\begin{equation}
  \mathbf J = \frac{\delta \mathscr L}{\delta\bar\psi} = -\bar\xi
\end{equation}
is the momentum map that takes $T^*D$ to $\mathfrak{sdiff}(D)^*$. This is employed to construct the partial Legendre transform $(\bar\psi,a) \mapsto (\bar\xi,a)$
\begin{equation}
  \mathscr H(\bar\xi,a) = \langle-\bar\xi,\bar\psi\rangle - \mathscr L(\bar\psi,a),
\end{equation}
which transforms the general Euler--Poincar\'e equations with advected quantities on $\mathfrak{sdiff}(D)\times \mathcal F(D)^n$ \eqref{eq:EP} into a Lie--Poisson system with Hamiltonian $\mathscr H(\bar\xi,a)$ on $\mathfrak{sdiff}(D)^*\ltimes \mathcal F(D)^n$. That is,
\begin{equation}
  \dot{\mathscr F} = \{\mathscr F, \mathscr H\} = \left\langle\bar\xi,\left[\frac{\mathscr F}{\delta\bar\xi},\frac{\mathscr H}{\delta\bar\xi}\right]\right\rangle +  \left\langle a,\left[\frac{\mathscr F}{\delta\bar\xi},\frac{\mathscr H}{\delta a}\right] + \left[\frac{\mathscr F}{\delta a},\frac{\mathscr H}{\delta\bar\xi}\right] \right\rangle
\end{equation}
for any $\mathscr F(\bar\xi,a)$. Specific dynamics, potentially more accurate than IL$^{(0,\alpha)}$QG dynamics, will emerge based on the Lagrangian $\mathscr{L}(\bar{\xi},a)$ and consequently the Hamiltonian $\mathscr{H}(\bar{\xi},a)$ chosen for a specific nonlocal dependence of $\bar{\psi}$ on $(\bar{\xi},a)$. This presents an opportunity for data science methods \cite{Brunton-Kutz-24} to facilitate the discovery of such dynamics.

\section{Summary and outlook}\label{sec:concl}

In this paper, we have demonstrated that the quasigeostrophic approximation of the recently proposed thermal rotating shallow-water equations with stratification can be derived from an Euler--Poincar\'e variational principle. These stratified thermal rotating shallow-water equations feature density (temperature) variations both horizontally and with depth in a polynomial manner, while maintaining the two-dimensional structure of the adiabatic rotating shallow-water equations.

Through the Euler--Poincar\'e variational formulation, we established a Kelvin--Noether theorem for the model, a result previously only known to exist for its primitive-equation counterpart. Furthermore, we illustrated that the model's Lie--Poisson Hamiltonian structure, previously derived via direct calculation, can also be obtained through a Legendre transform. This geometric nature is clarified by an appropriate momentum map.

In a noteworthy advance in our current understanding, we have identified a correspondence between Casimirs and Noether quantities of the Euler--Poincar\'e variational principle. This elucidates the explicit linkage of these conservation laws with the symmetry of the Lagrangian under particle relabeling.

The dual Euler--Poincar\'e/Lie--Poisson formalism offers a unified framework for describing quasigeostrophic stratified thermal flow, analogous to that for the primitive equations.

Future work should aim to establish an equivalent dual Euler--Poincar\'e/Lie--Poisson formulation for the quasigeostrophic equations governing fully three-dimensional, arbitrarily stratified flow. This effort seeks to derive the current model and, more importantly, to identify potential improvements through appropriate truncations of the Lagrangian.  Additionally, exploring the possibility of learning Euler--Poincar\'e/Lie--Poisson dynamics consistent with observational data using interpretable data science tools is also on the agenda.

\section*{Acknowledgments}

We dedicate this paper to the memory of Vladimir Zeitlin, whose work in thermal geophysical flow modeling has contributed to the resurgence of the field.

\section*{Funding}

The authors disclose no financial sponsorships or funding sources for this work.

\section*{Author declarations}

\subsection*{Conflict of interest}

The authors have no conflict of interest to disclose.

\subsection*{Author contributions}

The authors collectively conducted and composed this work.

\section*{Data availability}

This study did not incorporate any empirical data.

\bibliographystyle{alpha}

\end{document}